\documentclass[
aps,pra,twocolumn,10pt,
superscriptaddress,longbibliography,floatfix
]{revtex4-2}
\usepackage{amsmath,amssymb,amsthm,amsfonts}
\usepackage{microtype}
\usepackage{orcidlink}
\usepackage[dvipsnames]{xcolor}
\usepackage{graphicx}
\usepackage{url}
\usepackage{CJKutf8}
\usepackage{tikz}
\usepackage{tikz-3dplot}
\usetikzlibrary{arrows.meta}
\usetikzlibrary{positioning,arrows.meta,shapes.geometric}

\allowdisplaybreaks
\DeclareMathOperator{\Tr}{Tr}
\DeclareMathOperator{\diag}{diag}
\DeclareMathOperator{\rank}{rank}
\DeclareMathOperator{\Herm}{Herm}
\DeclareMathOperator{\Span}{Span}

\DeclareMathOperator{\im}{Im}
\DeclareMathOperator{\id}{id}

\newtheorem{theorem}{Theorem}[section]
\newtheorem{corollary}[theorem]{Corollary}
\newtheorem{lemma}[theorem]{Lemma}
\newtheorem{proposition}[theorem]{Proposition}
\newtheorem{definition}[theorem]{Definition}
\newtheorem{remark}[theorem]{Remark}
\newtheorem{example}[theorem]{Example}

\usepackage{hyperref}
\hypersetup{breaklinks=true, pdfborder={0 0 0}, colorlinks=true,
citecolor=magenta, linkcolor=blue, urlcolor=violet}

\begin{document}

\title{Operational Concealment of Measurement Incompatibility: Rank Loss versus Contraction}

\author{Mohd Asad Siddiqui\;\orcidlink{0000-0001-5003-7571}}
\email{asad@ctp-jamia.res.in}
\affiliation{Jawaharlal Nehru Rajkeeya Mahavidyalaya, Sri Vijaya Puram,
744104, Andaman and Nicobar Islands, India}

\author{Zizhu Wang \begin{CJK}{UTF8}{gbsn}(王子竹)\end{CJK}\;\orcidlink{0000-0003-1669-4548}}
\email{zizhu@uestc.edu.cn}
\affiliation{Institute of Fundamental and Frontier Sciences, University of Electronic Science and Technology of China, 611731 Chengdu, China.}
\affiliation{Key Laboratory of Quantum Physics and Photonic Quantum Information (University of Electronic Science and Technology of China), Ministry of Education, 611731, Chengdu, China}

\begin{abstract}
Quantum measurements can remain incompatible yet become impossible to certify when all probe states first pass through a quantum channel. We formulate this loss of certifiability as operational concealment and identify its exact finite-dimensional threshold. Under tomographically complete input access, a channel conceals some incompatible binary POVM pair if and only if its Heisenberg adjoint has nontrivial kernel. Thus, in this channel-restricted single-use setting, rank loss rather than contraction is the universal transition for exact concealment: injective dissipative dynamics may render the effective measurements compatible while leaving exact concealment impossible.
We quantify this distinction through a minimum-gain lower bound and show that regular time-local dynamics cannot produce exact concealment at finite time. We then classify all qubit channels whose adjoint fibres admit a positive unital retraction, reducing concealment to ordinary joint measurability of canonical representatives. For orthogonal projected qubit measurements we obtain an exact robustness formula and a strict separation from ordinary incompatibility robustness. Finally, quotient-space duality expresses this robustness as an experimentally accessible additive advantage in positive-payoff channel-output games. We also prove invariance under adjoining a passive finite-dimensional reference system for lifted target measurements, even when arbitrary joint compatible simulators are allowed on the enlarged output space. These results establish rank-sensitive criteria for certifying measurement incompatibility under restricted channel access.
\end{abstract}

\maketitle

\section{Introduction}
\label{sec:intro}

Measurement incompatibility---the impossibility of realizing multiple quantum measurements as marginals of a single joint measurement---is a fundamental nonclassical feature underlying steering and nonlocality \cite{Wolf2009,Heinosaari2016,Guhne2023}, and a resource for operational quantum-information advantages \cite{Uola2019,Buscemi2020}. Standard compatibility tests assume rich control over probe states, enabling complete tomography of the measurement statistics. In practice, however, preparation noise or intervening dynamics restricts the available test states to the image of a quantum channel. Even when the measurements remain mathematically incompatible, their incompatibility may then become operationally uncertifiable.

Imagine preparing a tomographically complete family of states, sending them through a noisy channel, and measuring the outputs. The central question is whether an incompatible measurement pair can become statistically indistinguishable from a compatible pair on every channel-accessible output state and, if so, what channel property enables such concealment.

This restriction leads to two operationally distinct problems. The first is whether the statistics of an incompatible POVM family on every channel output can be reproduced by a compatible family; we call this \emph{exact operational concealment}. The second is whether the Heisenberg-picture measurements $\mathcal E^\dagger(M_{a|x})$ themselves become compatible; we call this \emph{effective compatibility}. The distinction is structural: concealment concerns compatible representatives in affine fibres of $\mathcal E^\dagger$, whereas effective compatibility concerns the positive-unital action of $\mathcal E^\dagger$ on effects. The natural observable space for channel-restricted tests is the quotient $\Herm(\mathcal H_{\rm out})/\ker(\mathcal E^\dagger)$, whose elements collect effects indistinguishable on all channel outputs. This quotient formulation makes rank loss and contraction sharply distinguishable: the former changes affine fibres, while the latter can change joint measurability without creating hidden observable directions.

For static channels we use adjoint noninjectivity or rank deficiency; in dynamics, where a nontrivial kernel develops from an initially invertible map, we call this rank loss.

\subsection{Main results and contributions}

Our results establish a channel-specific structural and quantitative theory of operational concealment:

\begin{enumerate}
\item \textbf{Exact certifiability transition.}
Under tomographically complete input access, a finite-dimensional channel can conceal some incompatible binary measurement pair if and only if its adjoint is noninjective. Consequently, regular time-local dynamics cannot generate exact concealment at finite time, although contraction can make the effective measurements compatible. A minimum-gain bound quantifies this separation before rank loss. These conclusions are invariant under arbitrary passive finite-dimensional ancillas for lifted target measurements, including arbitrary joint compatible representatives on the enlarged output space.

\item \textbf{Canonical representatives and exact geometry.}
We characterize all qubit channels whose adjoint fibres admit a positive unital retraction. Whenever such a retraction exists, concealment reduces exactly to ordinary compatibility of canonical representatives. This gives an exact robustness formula for orthogonal projected qubit measurements and a strict separation from ordinary incompatibility robustness.

\item \textbf{Operational quotient theory.}
We derive the exact quotient dual, prove strong duality and attainment, and show that its witnesses are evaluable using only channel-output statistics. Concealment robustness equals an additive advantage in shift-invariant positive-payoff games.
\end{enumerate}

\paragraph*{Advance beyond prior work.}
Previous work established that compatibility relative to a restricted state set depends on the real linear span of the accessible states. The present contribution is not merely to specialize that observation to channel images. We identify adjoint noninjectivity as the universal structural condition enabling exact loss of incompatibility certifiability, prove that contraction while the adjoint remains injective cannot produce this effect, derive finite-time quantitative obstructions for regular dynamics, classify the qubit channels for which adjoint fibres admit positive canonical representatives, and develop an exact robustness and operational duality on the resulting quotient space. Incompatibility-breaking and preservability frameworks instead concern compatibility of the transformed measurements $\mathcal E^\dagger(M_{a|x})$ and therefore address a distinct operational layer.

\paragraph*{Scope and assumptions.}
We work in finite dimensions with finite-outcome measurements, exact channel knowledge, tomographically complete single-use input access, and exact statistical equivalence unless an approximation parameter is introduced explicitly. The universal noninjectivity criterion and the minimum-gain obstruction hold under these general assumptions. The positive-retraction reduction requires additional structure, and the closed robustness formula applies to orthogonal projected binary qubit measurements. Ancilla-assisted and finite-sample consequences are discussed separately.

Compatibility relative to restricted state sets was developed in Ref.~\cite{Heinosaari2021FewStates}. Subspace-restricted compatibility and compatibility dimension were studied in Refs.~\cite{Loulidi2021,Uola2021Subspaces}, while incompatibility in restricted operational theories was analysed in Ref.~\cite{Plavala2022}. The present channel setting adds the linear structure that the annihilator of the accessible state span is exactly $\ker(\mathcal E^\dagger)$.

Recent work on compatibility regions under decoherence \cite{McNulty2026}, incompatibility-breaking channels \cite{Heinosaari2015IB}, and the dynamical resource theory of incompatibility preservability \cite{Hsieh2025} concerns channel-level destruction or preservation of effective incompatibility. Exact concealment instead asks whether a specified original measurement pair possesses compatible representatives in its adjoint fibres. A recent classification under classical pre- and post-processing \cite{Das2026} provides a complementary non-channel comparison.

A single input state cannot certify incompatibility \cite{Heinosaari2021FewStates}: for any fixed output state, compatible proportional-to-identity POVMs can reproduce the statistics of any finite-outcome pair (Appendix~\ref{app:single_state}). We therefore assume tomographically complete input access unless explicitly stated otherwise.

Complete dephasing exemplifies the core mechanism: rank loss in the adjoint creates operational equivalence classes containing compatible representatives. The Pauli $X$ effects are operationally equivalent to the trivial effects $I/2$, while the Pauli $Z$ effects remain unchanged. Consequently, the incompatible $(X,Z)$ pair is indistinguishable from a compatible pair on every dephased output; see Fig.~\ref{fig:concealment_geometry}.

\section{Operational Framework}
\label{sec:framework}

In this section we present the mathematical background and standing assumptions, establish the fundamental adjoint characterization of concealment, and introduce the distinction between affine fibres and order-theoretic structures that underlies the entire analysis.

\subsection{Definitions and standing assumptions}

All Hilbert spaces in this work are finite-dimensional. Since $\mathcal{E}^\dagger$ is complex-linear and preserves Hermiticity, injectivity on the full operator space is equivalent to injectivity on the Hermitian subspace. We write $\ker(\mathcal E^\dagger)$ to denote the restriction of the adjoint kernel to the Hermitian operator space:
\[
\ker(\mathcal E^\dagger) = \{X \in \Herm(\mathcal H_{\mathrm{out}}) : \mathcal E^\dagger(X)=0\}.
\]

Positive-operator-valued measures (POVMs) $\{M_a\}_{a\in\Omega}$ satisfy $M_a \ge 0$, $\sum_a M_a = I$ \cite{Busch2016,Heinosaari2016,Guhne2023}. Two POVMs $\{M_a\}_{a\in\Omega}$ and $\{N_b\}_{b\in\Lambda}$ are \emph{compatible} if there exists a joint POVM $\{J_{ab}\}$ with $\sum_b J_{ab} = M_a$, $\sum_a J_{ab} = N_b$ \cite{Busch2016,Heinosaari2016,Guhne2023}.

For a convex set $C$, we denote by $\operatorname{ri}(C)$ its \emph{relative interior}, i.e., the interior of $C$ when $C$ is regarded as a subset of its affine hull.

\paragraph*{Notation.}
Throughout the paper, $(M,N)$ denotes the original POVM pair, $(F,G)$ compatible representatives of the operational equivalence classes, and $\{J_{ab}\}$ a joint POVM with marginals $\sum_bJ_{ab}=F_a$ and $\sum_aJ_{ab}=G_b$. The adjoint channel is denoted by $\mathcal E^\dagger$, and for any POVM $\{M_a\}$, the POVM $\{\mathcal E^\dagger(M_a)\}$ is called the corresponding effective measurement.

A quantum channel $\mathcal{E}: \mathcal{L}(\mathcal{H}_{\mathrm{in}}) \to \mathcal{L}(\mathcal{H}_{\mathrm{out}})$ is a completely positive and trace preserving (CPTP) map. Its adjoint $\mathcal{E}^\dagger: \mathcal{L}(\mathcal{H}_{\mathrm{out}}) \to \mathcal{L}(\mathcal{H}_{\mathrm{in}})$ is completely positive and unital, and is defined by
\[
\Tr[\mathcal{E}(\rho) A] = \Tr[\rho \,\mathcal{E}^\dagger(A)]
\]
for all states $\rho$ and observables $A$. This identity expresses the equivalence between the Schr\"odinger and Heisenberg pictures: evolving states forward under $\mathcal E$ is equivalent to evolving observables backward under $\mathcal E^\dagger$.

Consequently, every POVM $\{M_a\}$ on the output system induces the POVM $\{\mathcal{E}^\dagger(M_a)\}$ on the input system through the Heisenberg action of the channel. We shall refer to $\{\mathcal E^\dagger(M_a)\}$ as the \emph{effective measurement induced by} $\mathcal E$.

The kernel of the adjoint channel plays a central role in our analysis. It consists precisely of those hidden observable directions that are inaccessible through measurements performed after the channel. We say that a channel has an \emph{injective adjoint} if
\[
\mathcal{E}^\dagger(X) = 0 \Rightarrow X = 0
\]
for $X \in \mathcal{L}(\mathcal{H}_{\mathrm{out}})$. Equivalently, $\ker(\mathcal{E}^\dagger) = \{0\}$.

A set $\mathcal{T} \subseteq \mathcal{S}(\mathcal{H}_{\mathrm{in}})$ is \emph{tomographically complete} if its real linear span (over $\mathbb{R}$) equals the Hermitian operators on $\mathcal{H}_{\mathrm{in}}$ \cite{ParisRehacek2004}. Tomographic completeness allows equality of measurement statistics to be promoted to corresponding operator equalities involving the adjoint channel.

\begin{remark}[Standing assumptions]
\label{rem:assumptions}
Unless otherwise stated, we assume:
(i) finite-dimensional Hilbert spaces,
(ii) finite-outcome POVMs,
(iii) exact tomography, and
(iv) tomographic completeness of the input family $\mathcal T$.
All operator equalities are understood on the corresponding finite-dimensional Hermitian spaces.
\end{remark}

\subsection{Adjoint characterization}

\begin{definition}[Operational concealment]
\label{def:family}
Let $\mathcal{T} \subseteq \mathcal{S}(\mathcal{H}_{\mathrm{in}})$ be a set of input states. POVMs $\{M_a\},\{N_b\}$ are \emph{concealed} by a quantum channel $\mathcal{E}$ relative to $\mathcal{T}$ if there exist compatible POVMs $\{F_a\},\{G_b\}$ such that for all $\rho \in \mathcal{T}$,
$\Tr[F_a \mathcal{E}(\rho)] = \Tr[M_a \mathcal{E}(\rho)]$ and $\Tr[G_b \mathcal{E}(\rho)] = \Tr[N_b \mathcal{E}(\rho)]$.
\end{definition}

In this work, our primary focus is the case where $\mathcal{T}$ is tomographically complete. Under this assumption, operational concealment admits an exact characterization in terms of the adjoint channel.

\begin{theorem}[Adjoint characterization of concealment]
\label{thm:necessary_sufficient}
Under the assumptions of Remark~\ref{rem:assumptions}, $\{M_a\},\{N_b\}$ are concealed by $\mathcal{E}$ if and only if there exist compatible POVMs $\{F_a\},\{G_b\}$ with $\mathcal{E}^\dagger(F_a) = \mathcal{E}^\dagger(M_a)$ and $\mathcal{E}^\dagger(G_b) = \mathcal{E}^\dagger(N_b)$ for every outcome $a$ and $b$.
\end{theorem}

\begin{proof}
\emph{Necessity.} Concealment gives
\begin{equation}
\Tr[\rho\,\mathcal E^\dagger(F_a-M_a)]=0,
\qquad \rho\in\mathcal T.
\end{equation}
Since $F_a-M_a$ is Hermitian and $\mathcal{E}^\dagger$ preserves Hermiticity, $\mathcal{E}^\dagger(F_a-M_a)$ is Hermitian. By tomographic completeness, $\operatorname{span}_{\mathbb{R}}(\mathcal{T}) = \Herm(\mathcal{H}_{\mathrm{in}})$. Since the Hilbert-Schmidt inner product is nondegenerate on $\Herm(\mathcal{H}_{\mathrm{in}})$, it follows that $\mathcal{E}^\dagger(F_a-M_a) = 0$.

\emph{Sufficiency.} If $\mathcal{E}^\dagger(F_a) = \mathcal{E}^\dagger(M_a)$, then
\begin{align*}
\Tr[F_a\mathcal E(\rho)]
&=\Tr[\rho\,\mathcal E^\dagger(F_a)]\\
&=\Tr[\rho\,\mathcal E^\dagger(M_a)]
=\Tr[M_a\mathcal E(\rho)].
\end{align*}
The same argument gives $\Tr[G_b\mathcal{E}(\rho)] = \Tr[N_b\mathcal{E}(\rho)]$ for all $\rho$. Hence the compatible POVMs $\{F_a\},\{G_b\}$ reproduce the statistics of $\{M_a\},\{N_b\}$, establishing concealment.
\end{proof}

Since the equality holds on the tomographically complete set $\mathcal T$, linearity extends it to the real span of $\mathcal T$, namely $\Herm(\mathcal H_{\mathrm{in}})$, and hence to all density operators.

\subsection{Affine fibres and effective measurements}

A central distinction in this work is between two different mathematical structures that appear in the analysis:

\begin{enumerate}
\item \textbf{Affine geometry:} The vector space $\Herm(\mathcal H_{\mathrm{out}})$ and its quotient by $\ker(\mathcal E^\dagger)$. This governs exact concealment---whether compatible representatives exist in affine fibres. Only the linear structure matters.

\item \textbf{Order theory:} The positive cone $\mathcal P(\mathcal H) = \{X \in \Herm(\mathcal H) : X \ge 0\}$ and its image under $\mathcal E^\dagger$. This governs effective compatibility---whether $\mathcal E^\dagger(M)$ and $\mathcal E^\dagger(N)$ are jointly measurable on the input space. The full positive unital structure of $\mathcal E^\dagger$ matters.
\end{enumerate}

Exact concealment asks: does the affine fibre $M + \ker(\mathcal E^\dagger)$ intersect the set of effects that are marginals of a joint POVM? Effective compatibility asks: are $\mathcal E^\dagger(M)$ and $\mathcal E^\dagger(N)$ jointly measurable on $\mathcal H_{\mathrm{in}}$?

This distinction explains why the two layers are inequivalent. We will return to this point in the Discussion.

Theorem~\ref{thm:necessary_sufficient} induces the equivalence relation
\[
A\sim B
\quad\Longleftrightarrow\quad
A-B\in\ker(\mathcal E^\dagger).
\]
Because the adjoint kernel is a linear subspace, $\sim$ is an equivalence relation. Under tomographically complete input access, $A\sim B$ holds exactly when $A$ and $B$ produce identical statistics on every channel output. The operationally accessible observable space is therefore the quotient
\[
\mathsf Q_{\mathcal E}
:=
\Herm(\mathcal H_{\rm out})/\ker(\mathcal E^\dagger).
\]
The First Isomorphism Theorem gives the natural vector-space isomorphism
\begin{equation}
\mathsf Q_{\mathcal E}
\cong
\im(\mathcal E^\dagger),
\qquad
[M]\longmapsto\mathcal E^\dagger(M).
\label{eq:isomorphism}
\end{equation}
For a fixed outcome set, an operational POVM is a tuple of quotient classes admitting normalized positive representatives. These tuples form a convex set. We denote the canonical quotient map, extended componentwise to POVMs, by
\[
\pi:\Herm(\mathcal H_{\rm out})\to\mathsf Q_{\mathcal E},
\qquad M\mapsto[M].
\]

\begin{figure}[tbp]
\centering
\begin{tikzpicture}[
every node/.style={font=\small,align=center},
box/.style={
draw,
rounded corners,
minimum width=2.8cm,
minimum height=1.0cm,
inner sep=4pt
}
]

\node[box] (M) at (0,0)
{Original\\POVM pair};

\node[box] (Eff) at (4.2,0)
{Effective\\measurements};

\node[box] (Fib) at (0,-2.4)
{Operational\\equivalence\\classes};

\draw[->,thick]
(M) -- node[above] {$\mathcal E^\dagger$} (Eff);

\draw[->,thick]
(M) -- node[left] {$\pi$} (Fib);

\draw[->,thick]
(Fib) -- node[pos=0.55,right,font=\scriptsize] {$\overline{\mathcal E^\dagger}$} (Eff);

\end{tikzpicture}

\caption{
Conceptual structure of operational concealment.
The quotient map $\pi$ sends a POVM pair to its operational equivalence classes $[M_a]=M_a+\ker(\mathcal E^\dagger)$ and $[N_b]=N_b+\ker(\mathcal E^\dagger)$. The induced map $\overline{\mathcal E^\dagger}$ identifies the quotient with $\im(\mathcal E^\dagger)$. Concealment occurs precisely when the two classes admit compatible representatives; it implies compatibility of the effective measurements, but not conversely.
}
\label{fig:concealment_diagram}
\end{figure}
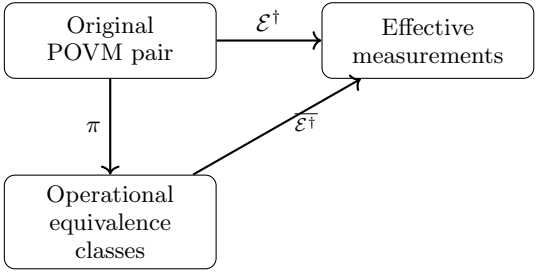

Fig.~\ref{fig:concealment_diagram} summarizes the logical relation between the original POVMs, their operational equivalence classes, and the effective measurements. The lower arrow is the induced quotient map $\overline{\mathcal E^\dagger}$ from Eq.~\eqref{eq:isomorphism}; the concealment condition is the existence of compatible representatives within the two operational fibres.

\subsection{Relation to restricted-state compatibility}

\begin{proposition}[Concealment as restricted-state compatibility]
\label{prop:restricted_state}
Under the standing assumptions of Remark~\ref{rem:assumptions}, let
\[
\mathcal S_{\mathcal E}:=\{\mathcal E(\rho):\rho\in\mathcal S(\mathcal H_{\rm in})\}
\]
be the channel-accessible state set. Then $\{M_a\},\{N_b\}$ are concealed by $\mathcal E$ if and only if they are $\mathcal S_{\mathcal E}$-compatible, i.e., there exist compatible POVMs $\{F_a\},\{G_b\}$ such that
\begin{align*}
\Tr[F_a\sigma]&=\Tr[M_a\sigma],\\
\Tr[G_b\sigma]&=\Tr[N_b\sigma],
\qquad \sigma\in\mathcal S_{\mathcal E}.
\end{align*}
\end{proposition}

\begin{proof}
If the pair is concealed relative to $\mathcal T$, Theorem~\ref{thm:necessary_sufficient} gives $\mathcal E^\dagger(F_a-M_a)=0$ and $\mathcal E^\dagger(G_b-N_b)=0$. Hence the statistical equalities hold for every input state and therefore for every $\sigma\in\mathcal S_{\mathcal E}$. Conversely, $\mathcal S_{\mathcal E}$-compatibility implies the defining equalities on $\{\mathcal E(\rho):\rho\in\mathcal T\}$.
\end{proof}

\section{Rank-Loss Classification, Approximate Bounds, and Dynamical Separation}
\label{sec:classification}

We now ask when channel noise can exactly hide at least one incompatible measurement pair. The answer is a nontrivial adjoint kernel, corresponding to an observable direction invisible on all channel outputs. We then derive approximate bounds and channel ordering, contrast concealment with effective compatibility, and establish a finite-time dynamical obstruction.

\subsection{Universal noninjectivity criterion}

\begin{corollary}[Injective adjoint forbids concealment]
\label{cor:invertible}
If $\mathcal{E}^\dagger$ is injective (equivalently, $\ker(\mathcal{E}^\dagger)=\{0\}$), then a POVM pair can be concealed by $\mathcal E$ only if it is already compatible.
\end{corollary}
\begin{proof}
By Theorem~\ref{thm:necessary_sufficient}, concealment would imply the existence of compatible POVMs $\{F_a\},\{G_b\}$ satisfying $\mathcal{E}^\dagger(F_a) = \mathcal{E}^\dagger(M_a)$ and $\mathcal{E}^\dagger(G_b) = \mathcal{E}^\dagger(N_b)$. Injectivity implies $F_a = M_a$ and $G_b = N_b$ for all outcomes. Hence the compatible representatives coincide with the original POVMs. Therefore concealment is possible only if $\{M_a\}$ and $\{N_b\}$ are already compatible.
\end{proof}

The adjoint characterization immediately gives a necessary condition: if the adjoint is injective, every operational equivalence class is a singleton, so concealment can only reproduce compatibility already present in the original pair. The converse is less obvious. It requires showing that every nonzero hidden Hermitian direction can be converted into an incompatible binary POVM that is operationally equivalent to a trivial one. Together, the two directions yield the central structural result of the paper.

\begin{theorem}[Universal concealment criterion]
\label{thm:universal_concealment}
Under the standing tomographic-completeness assumption (Remark~\ref{rem:assumptions}), let $\mathcal E$ be a finite-dimensional quantum channel. Then $\mathcal E$ conceals, relative to any tomographically complete input family, at least one incompatible pair of binary POVMs if and only if
\[
\ker(\mathcal E^\dagger)\neq\{0\}.
\]
\end{theorem}

\begin{proof}
The forward implication is Corollary~\ref{cor:invertible}: if $\mathcal E^\dagger$ is injective, no incompatible pair can be concealed.

For the converse, suppose $\ker(\mathcal E^\dagger)$ contains a nonzero Hermitian operator $H$. Since $\mathcal E^\dagger$ is unital, $\mathcal E^\dagger(I)=I$, so $H$ cannot be a nonzero scalar multiple of $I$. Thus $H$ is not proportional to the identity.

Choose $\epsilon>0$ such that $\epsilon \le (2\|H\|_\infty)^{-1}$ and define
\[
M_\pm = \frac{I}{2} \pm \epsilon H.
\]
Then $M_\pm\ge0$ and $M_++M_-=I$, so $\{M_+,M_-\}$ is a valid binary POVM. Since $H\in\ker(\mathcal E^\dagger)$,
\[
\mathcal E^\dagger(M_\pm) = \mathcal E^\dagger\left(\frac{I}{2}\right),
\]
so $M$ is operationally equivalent to the trivial POVM $F_\pm=I/2$.

Because $H$ is not proportional to $I$, there exists a rank-one projection $P$ such that $[P,H]\neq0$. Indeed, choose two eigenvectors \(|h_i\rangle, |h_j\rangle\) of \(H\) with distinct eigenvalues, set
\[
|\psi\rangle = \frac{|h_i\rangle + |h_j\rangle}{\sqrt2}, \qquad P = |\psi\rangle\langle\psi|.
\]
Then
\[
\langle h_i | [P,H] | h_j \rangle = \frac{\lambda_j - \lambda_i}{2} \neq 0,
\]
so \([P,H] \neq 0\). Let \(N=\{P,I-P\}\) be the corresponding sharp binary POVM.

We claim that $M$ and $N$ are incompatible. Suppose, for contradiction, that they are jointly measurable. Then there exists a joint POVM $\{J_{\alpha b}\}_{\alpha=\pm,b=1,2}$ with marginals
\[
\sum_{b=1}^2 J_{\alpha b}=M_\alpha,\qquad
\sum_{\alpha=\pm} J_{\alpha b}=P_b,
\]
where $P_1=P$ and $P_2=I-P$. Since $P_b$ are projections, the marginal constraints imply
\[
0\le J_{\alpha b}\le P_b,
\]
so $J_{\alpha b}=P_b J_{\alpha b}P_b$. Hence each $J_{\alpha b}$ commutes with $P$, and therefore
\[
M_\alpha=\sum_b J_{\alpha b}
\]
also commutes with $P$. But $M_\alpha=\frac I2\pm\epsilon H$, so $[M_\alpha,P]=\pm\epsilon[H,P]\neq0$, a contradiction. Thus $M$ and $N$ are incompatible.

Finally, the pair $(F,N)$ is compatible because $F$ is trivial: the joint POVM \(J_{\alpha b} = \frac12 P_b\) has marginals \(F_\alpha\) and \(P_b\). By Theorem~\ref{thm:necessary_sufficient}, $(M,N)$ is concealed by $\mathcal E$.
\end{proof}

For unital qubit channels, the canonical projection of Section~\ref{sec:retractions} gives an explicit Bloch-space realization of the universal criterion.

\begin{remark}[Role of tomographic completeness]
Without tomographic completeness, the theorem is false. For example, let $\mathcal E=\id$ be injective, but allow testing only on a single state. Then incompatible measurements can be indistinguishable on that state, so an injective channel can conceal incompatibility relative to a non-tomographically-complete test family.
\end{remark}

\begin{figure*}[tbp]
\centering
\includegraphics[width=0.9\linewidth]{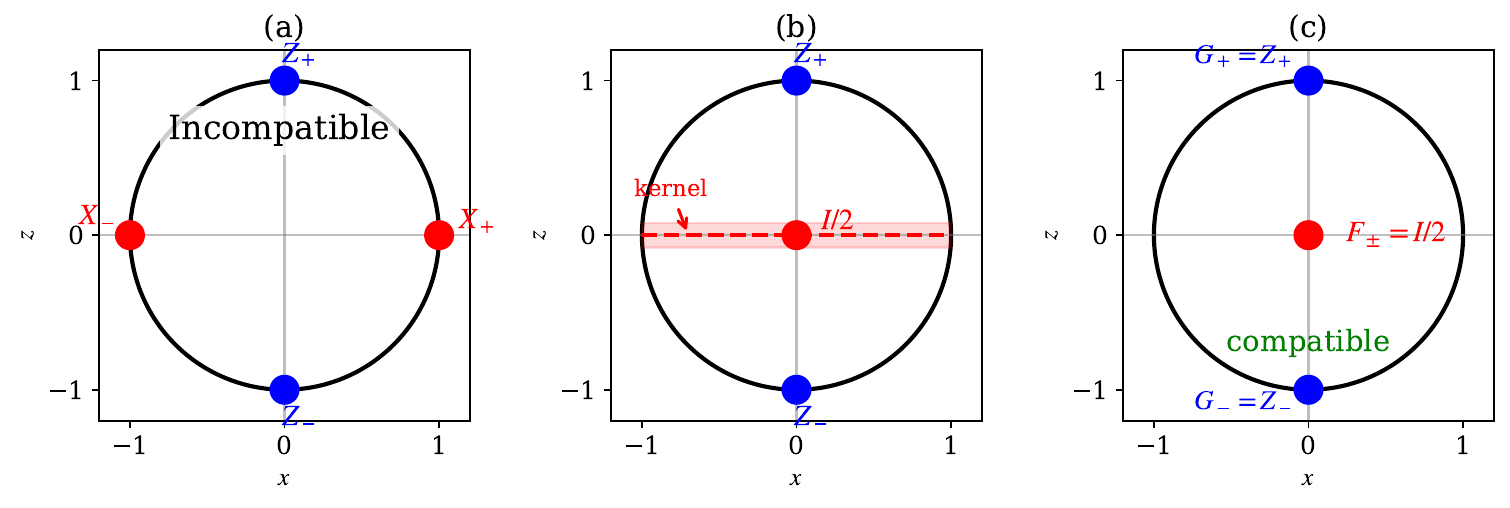}
\caption{
(Color online) Operational concealment under complete dephasing. Although projective POVMs $\{X_+,X_-\}$ and $\{Z_+,Z_-\}$ are incompatible, the Pauli $X$ effects become operationally equivalent to $I/2$ after dephasing, while Pauli $Z$ remains unchanged. The pair $(\{I/2,I/2\},\{Z_+,Z_-\})$ is compatible, so the incompatibility is operationally concealed. Physically, complete dephasing erases the $X$-observable direction from the channel-output statistics, so the $X$-$Z$ incompatibility cannot be certified within this access model.
}
\label{fig:concealment_geometry}
\end{figure*}

Theorem~\ref{thm:universal_concealment} identifies the adjoint kernel as the exact invariant governing the \emph{possibility} of concealment under tomographically complete access. In particular, the existence question is insensitive to the nonzero singular values or contraction strengths of the channel: only the presence of a hidden observable direction matters. For a prescribed measurement pair, however, noninjectivity alone is not sufficient; the corresponding affine fibres must still intersect the compatible set. This distinction between an existence criterion and a pair-dependent feasibility problem motivates the structural and quantitative analysis below.

\subsection{Quantitative obstruction to approximate concealment}

Let $\mathrm{JM}$ denote the set of jointly measurable POVM pairs with the same outcome sets. Define the distance of a POVM pair from compatibility:
\begin{equation}
\begin{aligned}
d_{\mathrm{JM}}(M,N)
:=
\inf_{(F,G)\in\mathrm{JM}}
\max\Bigl\{&
\max_a\|M_a-F_a\|_\infty,\\
&
\max_b\|N_b-G_b\|_\infty
\Bigr\}.
\end{aligned}
\end{equation}

For an injective adjoint, define its minimum gain
\[
\mu_\infty(\mathcal E^\dagger)
:=
\inf_{\substack{
X\in\Herm(\mathcal H_{\rm out})\\
\|X\|_\infty=1}}
\|\mathcal E^\dagger(X)\|_\infty.
\]
Equivalently,
\[
\mu_\infty(\mathcal E^\dagger)
=
\frac{1}{
\left\|
(\mathcal E^\dagger)^{-1}_{\im}
\right\|_{\infty\to\infty}
},
\]
where
\[
(\mathcal E^\dagger)^{-1}_{\im}:
\im(\mathcal E^\dagger)
\longrightarrow
\Herm(\mathcal H_{\rm out})
\]
denotes the inverse of $\mathcal E^\dagger$ on its image. Here \(\|\cdot\|_{\infty\to\infty}\) denotes the induced operator norm associated with the spectral norm on both domain and codomain.

\begin{theorem}[Quantitative obstruction to approximate concealment]
\label{thm:quantitative_obstruction}
Let $\mathcal E^\dagger$ be injective and let $\epsilon_c(M,N\mid\mathcal E)$ denote the operator-norm approximate-concealment error defined in Appendix~\ref{app:approx}. Then
\[
\mu_\infty(\mathcal E^\dagger)\, d_{\rm JM}(M,N)
\le
\epsilon_c(M,N\mid\mathcal E)
\le
d_{\rm JM}(M,N).
\]
\end{theorem}

\begin{proof}
Let $(F,G)$ be compatible. Since $\mathcal E^\dagger$ is unital and positive, it is contractive in operator norm, so
\[
\epsilon_c(M,N\mid\mathcal E)\le d_{\rm JM}(M,N).
\]

For the lower bound, for any compatible $(F,G)$,
\[
\|\mathcal E^\dagger(M_a-F_a)\|_\infty
\ge
\mu_\infty(\mathcal E^\dagger)\|M_a-F_a\|_\infty,
\]
and similarly for $N_b-G_b$. Taking the maximum over outcomes and then the infimum over compatible pairs gives the stated inequality.
\end{proof}

If $\mathcal E_t$ satisfies the time-local equation $\dot{\mathcal E}_t=\mathcal L_t\mathcal E_t$, with $\mathcal E_0=\id$ and locally integrable $\mathcal L_t$, then
\[
\|(\mathcal E_t^\dagger)^{-1}\|_{\infty\to\infty}
\le
\exp\left(\int_0^t \|\mathcal L_s^\dagger\|_{\infty\to\infty}\,ds\right),
\]
so
\[
\epsilon_c(M,N\mid\mathcal E_t)
\ge
\exp\left(-\int_0^t \|\mathcal L_s^\dagger\|_{\infty\to\infty}\,ds\right)
d_{\rm JM}(M,N).
\]

This upgrades the finite-time no-go theorem: for every finite time while the channel remains injective, there is a positive lower bound on the approximate-concealment error controlled by the minimum gain of the channel. In particular, for the depolarizing semigroup, this bound remains positive for all finite $t$, even when $t\ge t_{\rm eff}$ and the effective measurements have become compatible.

\subsection{Kernel invariance and channel ordering}

Having established when concealment is possible, we now study the geometry of the concealed set and its dependence on the channel. For fixed finite outcome sets $\Omega$ and $\Lambda$, let $\mathcal C_{\mathcal E}^{\Omega,\Lambda}$ denote the set of POVM pairs concealed by $\mathcal E$. When the outcome sets are understood, we write simply $\mathcal C_{\mathcal E}$.

\begin{proposition}[Convexity and compactness of $\mathcal C_{\mathcal E}$]
\label{prop:convexity_compactness}
For fixed finite outcome sets $\Omega$ and $\Lambda$, the set $\mathcal C_{\mathcal E}^{\Omega,\Lambda}$ is convex and compact.
\end{proposition}

\begin{proof}
The proof is provided in Appendix~\ref{app:convexity_compactness}.
\end{proof}

We next examine how operational concealment depends on the underlying channel. The following results show that the adjoint kernel provides a natural invariant for operational concealment.

\begin{theorem}[Kernel invariance of concealment]
\label{thm:kernel_universality}
Under the assumptions of Theorem~\ref{thm:necessary_sufficient}, let $\mathcal{E}_1$ and $\mathcal{E}_2$ be quantum channels with common output Hilbert space $\mathcal{H}_{\mathrm{out}}$ satisfying $\ker(\mathcal{E}_1^\dagger) = \ker(\mathcal{E}_2^\dagger)$. Then $\mathcal{C}_{\mathcal{E}_1} = \mathcal{C}_{\mathcal{E}_2}$; i.e., a POVM pair is concealed by $\mathcal{E}_1$ if and only if it is concealed by $\mathcal{E}_2$.
\end{theorem}

\begin{proof}
By Theorem~\ref{thm:necessary_sufficient}, concealment under $\mathcal{E}_i$ is equivalent to the existence of compatible POVMs $\{F_a\},\{G_b\}$ satisfying
\[
\mathcal{E}_i^\dagger(F_a) = \mathcal{E}_i^\dagger(M_a), \qquad \mathcal{E}_i^\dagger(G_b) = \mathcal{E}_i^\dagger(N_b).
\]
Equivalently,
\[
F_a - M_a \in \ker(\mathcal{E}_i^\dagger), \qquad G_b - N_b \in \ker(\mathcal{E}_i^\dagger).
\]
If $\ker(\mathcal{E}_1^\dagger) = \ker(\mathcal{E}_2^\dagger)$, then the admissible operational equivalence classes coincide for both channels, and therefore the same concealment constructions exist. Hence $\mathcal{C}_{\mathcal{E}_1} = \mathcal{C}_{\mathcal{E}_2}$.
\end{proof}

Equality of adjoint kernels is sufficient for concealment-equivalence: channels with different dynamical realizations but identical adjoint kernels are indistinguishable with respect to concealment of POVM pairs. This is essentially the dual-space version of the known result that restricted-state compatibility depends only on the real linear span of the accessible state set~\cite{Heinosaari2021FewStates}. Whether equality of adjoint kernels is also necessary for concealment-equivalence remains open. Consequently, the adjoint kernel provides a sufficient invariant for the concealment framework developed here.

\begin{theorem}[Kernel monotonicity of concealment sets]
\label{thm:kernel_monotonicity}
Let $\mathcal{E}_1$ and $\mathcal{E}_2$ be quantum channels with the same output Hilbert space $\mathcal{H}_{\mathrm{out}}$. If $\ker(\mathcal{E}_1^\dagger)\subseteq\ker(\mathcal{E}_2^\dagger)$, then $\mathcal{C}_{\mathcal{E}_1}\subseteq\mathcal{C}_{\mathcal{E}_2}$.
\end{theorem}

\begin{proof}
Let $(M,N)\in\mathcal{C}_{\mathcal{E}_1}$. Then there exist compatible representatives $(F,G)$ with $F_a-M_a\in\ker(\mathcal{E}_1^\dagger)$ and $G_b-N_b\in\ker(\mathcal{E}_1^\dagger)$. Since $\ker(\mathcal{E}_1^\dagger)\subseteq\ker(\mathcal{E}_2^\dagger)$, the same differences lie in $\ker(\mathcal{E}_2^\dagger)$, so $(M,N)$ is concealed by $\mathcal{E}_2$.
\end{proof}

The kernel inclusion relation defines a preorder on quantum channels:
\[
\mathcal E_1\preceq_c\mathcal E_2 \quad\Longleftrightarrow\quad \ker(\mathcal E_1^\dagger)\subseteq \ker(\mathcal E_2^\dagger).
\]
Since distinct quantum channels may possess identical adjoint kernels, the relation is generally not antisymmetric and therefore defines only a preorder. Passing to equivalence classes of channels with identical adjoint kernels yields a partial order.

\subsection{Effective compatibility versus concealment}

We now turn to effective compatibility. Unlike the previous section, which characterizes operational concealment, this section studies the compatibility of the channel images \(\{\mathcal{E}^\dagger(M_a)\}\) and \(\{\mathcal{E}^\dagger(N_b)\}\). This distinction is essential because effective compatibility can occur even when operational concealment is impossible.

\begin{definition}[Effective compatibility]
Two POVMs $\{M_a\},\{N_b\}$ are \emph{effectively compatible} under $\mathcal E$ if $\{\mathcal E^\dagger(M_a)\}$ and $\{\mathcal E^\dagger(N_b)\}$ are jointly measurable on $\mathcal H_{\mathrm{in}}$.
\end{definition}

For unbiased binary qubit POVMs, we write measurement effects as
\begin{equation}
M_{\pm|0} = \frac{I \pm \mathbf{a}\cdot\boldsymbol{\sigma}}{2}, \qquad
M_{\pm|1} = \frac{I \pm \mathbf{b}\cdot\boldsymbol{\sigma}}{2},
\label{eq:binary_qubit_povm}
\end{equation}
with Bloch vectors $\mathbf{a},\mathbf{b} \in \mathbb{R}^3$ satisfying $|\mathbf{a}|,|\mathbf{b}| \leq 1$.
A necessary and sufficient criterion for the joint measurability of unbiased binary qubit POVMs is~\cite{Busch1986,Yu2010}:
\begin{equation}
|\mathbf{a}+\mathbf{b}| + |\mathbf{a}-\mathbf{b}| \leq 2.
\label{eq:jm_condition}
\end{equation}
This criterion will be used extensively in our qubit analysis.

\begin{proposition}[Concealment implies effective compatibility]
\label{prop:distinction}
If $\{M_a\},\{N_b\}$ are concealed by $\mathcal{E}$, then the effective POVMs $\{\mathcal{E}^\dagger(M_a)\}$ and $\{\mathcal{E}^\dagger(N_b)\}$ are compatible.
The converse fails in general: effective compatibility can occur for channels with injective adjoint, where concealment is impossible by Theorem~\ref{thm:universal_concealment}.
Thus concealment is strictly more restrictive than effective compatibility.
\end{proposition}

\begin{proof}
Since $\mathcal{E}^\dagger$ is completely positive and unital, $\{\mathcal{E}^\dagger(M_a)\}$ and $\{\mathcal{E}^\dagger(N_b)\}$ are POVMs. If $\{M_a\},\{N_b\}$ are concealed, Theorem~\ref{thm:necessary_sufficient} gives compatible $\{F_a\},\{G_b\}$ with $\mathcal{E}^\dagger(F_a)=\mathcal{E}^\dagger(M_a)$, $\mathcal{E}^\dagger(G_b)=\mathcal{E}^\dagger(N_b)$. Let $\{J_{ab}\}$ be a joint POVM for $\{F_a\},\{G_b\}$. Define $K_{ab} := \mathcal{E}^\dagger(J_{ab})$. Complete positivity of $\mathcal{E}^\dagger$ implies $K_{ab} \succeq 0$, while unitality gives
\[
\sum_{ab} K_{ab} = \mathcal{E}^\dagger\!\left(\sum_{ab} J_{ab}\right) = \mathcal{E}^\dagger(I_{\rm out}) = I_{\rm in}.
\]
Moreover, $\sum_b K_{ab} = \sum_b \mathcal{E}^\dagger(J_{ab}) = \mathcal{E}^\dagger(F_a) = \mathcal{E}^\dagger(M_a)$, and similarly $\sum_a K_{ab} = \mathcal{E}^\dagger(N_b)$. Hence $\{K_{ab}\}$ is a joint POVM for $\{\mathcal{E}^\dagger(M_a)\}$ and $\{\mathcal{E}^\dagger(N_b)\}$, establishing compatibility.

For the converse, consider binary qubit POVMs with $|\mathbf{a}+\mathbf{b}| + |\mathbf{a}-\mathbf{b}| > 2$, so the pair is incompatible. The depolarizing channel example below shows that the effective measurements can become compatible while concealment remains impossible.

For any $p_c\le p<1$, the adjoint depolarizing map is injective. Indeed, if
\[
\mathcal E_p^\dagger(A)
=
(1-p)A+\frac p2\Tr(A)I
=
0,
\]
taking the trace gives $\Tr(A)=0$, and hence
$(1-p)A=0$. Since $p<1$, it follows that $A=0$.
Therefore the effective pair is compatible while the original
incompatible pair cannot be concealed, by
Corollary~\ref{cor:invertible}.
\end{proof}

This implication links the two structures introduced in Section~\ref{sec:framework}: concealment concerns compatible representatives in adjoint fibres, whereas effective compatibility concerns joint measurability of their channel images. The converse fails when the channel contracts observable directions without annihilating them.

The following theorem applies the joint-measurability criterion to depolarizing noise. Related channel-level incompatibility-breaking conditions and bounds for depolarizing channels were studied in Ref.~\cite{Heinosaari2015IB}.

\begin{theorem}[Effective compatibility threshold under depolarizing noise]
\label{thm:depol_threshold}
Let $\mathcal M=\{M_{\pm|0},M_{\pm|1}\}$ be two unbiased binary qubit POVMs with Bloch vectors $\mathbf a,\mathbf b$. Under the depolarizing channel
\[
\mathcal E_p(\rho)=(1-p)\rho+p\frac{I}{2},
\qquad p\in[0,1],
\]
if $\mathbf a=\mathbf b=\mathbf 0$, the effective pair is compatible for every $p$. Otherwise, let
\[
L=|\mathbf a+\mathbf b|+|\mathbf a-\mathbf b|>0.
\]
Then the effective measurements are jointly measurable if and only if
\[
p\ge p_c
=
\max\!\left(0,1-\frac{2}{L}\right).
\]
\end{theorem}

\begin{proof}
The adjoint depolarizing channel acts as $\mathcal E_p^\dagger(A) = (1-p)A+\frac{p}{2}\Tr(A)I$. The effective Bloch vectors become $(1-p)\mathbf{a}$ and $(1-p)\mathbf{b}$. The joint-measurability criterion Eq.~\eqref{eq:jm_condition} gives $(1-p)L \le 2$, yielding the stated threshold.
\end{proof}

For sharp orthogonal measurements, $p_c = 1 - 1/\sqrt{2} \approx 0.292893$.

The following singular-value theorem captures the contraction effect for unital qubit channels in general. The condition characterizes when every sharp orthogonal unbiased binary pair is effectively compatible. For non-orthogonal or non-sharp pairs, the threshold depends on the specific Bloch vectors.

\begin{theorem}[Singular-value characterization for unital qubit channels]
\label{thm:singular_value}
Let $\mathcal E$ be a unital qubit channel with Bloch matrix $T$ and singular values $\sigma_1 \ge \sigma_2 \ge \sigma_3$. For a fixed pair of unbiased binary qubit POVMs with Bloch vectors $\mathbf{m},\mathbf{n}$, effective compatibility holds exactly when
\[
|T^{\mathsf T}\mathbf{m} + T^{\mathsf T}\mathbf{n}| + |T^{\mathsf T}\mathbf{m} - T^{\mathsf T}\mathbf{n}| \le 2.
\]
Moreover, every sharp orthogonal pair is effectively compatible iff
\[
\sigma_1^2 + \sigma_2^2 \le 1.
\]
\end{theorem}

The proof is provided in Appendix~\ref{app:technical}.

For comparison, observable--channel compatibility---whether an observable and a channel can arise from a common instrument---has been characterized for unbiased qubit observables and Pauli channels in Ref.~\cite{Heinosaari2018}. This information--disturbance notion is different from both effective compatibility and operational concealment considered here.

\subsection{Finite-time dynamical obstruction}

We now regard the channel as a time-dependent dynamical map $\mathcal E_t$ with $\mathcal E_0=\id$. The two-layer distinction has a sharp dynamical consequence: ordinary finite-time dissipation may make the effective measurements compatible, whereas exact concealment of an incompatible pair requires loss of linear rank.

\begin{theorem}[Finite-time no-go for regular dynamics]
\label{thm:regular_no_go}
Let $\mathcal E_t$ solve the finite-dimensional time-local equation
\begin{equation}
\frac{d}{dt}\mathcal E_t=\mathcal L_t\mathcal E_t,
\qquad \mathcal E_0=\id,
\label{eq:time_local}
\end{equation}
where $t\mapsto\mathcal L_t$ is a locally integrable finite-dimensional linear superoperator. Then $\mathcal E_t$ and $\mathcal E_t^\dagger$ are invertible linear maps for every finite $t$. Consequently, no incompatible finite-outcome POVM pair can be exactly concealed at finite time.
\end{theorem}

\begin{proof}
In finite dimension, Liouville's determinant formula gives
\[
\det\mathcal E_t
=
\exp\left[
\int_0^t
\operatorname{Tr}_{\rm sup}(\mathcal L_s)\,ds
\right],
\]
where $\operatorname{Tr}_{\rm sup}$ denotes the trace of the superoperator, with $\det\mathcal E_0=1$. Hence $\det\mathcal E_t\neq0$ at every finite time whenever the generator is locally integrable. Therefore $\mathcal E_t$ is invertible at every finite time. The adjoint of an invertible linear map is invertible, with
\[
(\mathcal E_t^\dagger)^{-1}
=
(\mathcal E_t^{-1})^\dagger.
\]
Hence $\ker(\mathcal E_t^\dagger)=\{0\}$. Corollary~\ref{cor:invertible} then excludes exact concealment of any incompatible pair.
\end{proof}

The theorem does not assert that $\mathcal E_t^{-1}$ is positive or completely positive; only linear invertibility is required. It applies, in particular, to finite-dimensional time-homogeneous quantum dynamical semigroups $\mathcal E_t=e^{t\mathcal L}$ at every finite time \cite{Gorini1976,Lindblad1976,Breuer2007}. Finite-time noninvertibility and singularities of dynamical maps have been studied in Refs.~\cite{Hou2012,Chruscinski2018}, while the associated singular behaviour and limitations of first-order time-local master equations were analysed in Ref.~\cite{Hegde2021}. The present theorem identifies a new operational consequence of such rank loss: it is necessary for finite-time exact concealment of an incompatible pair.

\subsubsection*{Depolarizing semigroup}

Consider, for $\gamma > 0$,
\begin{equation}
\mathcal D_t(\rho)=\eta_t\rho+(1-\eta_t)\frac{I}{2},
\qquad \eta_t=e^{-\gamma t}.
\label{eq:depolarizing_semigroup}
\end{equation}
For orthogonal Pauli measurements $X$ and $Z$, the effective Bloch vectors are $\eta_t\mathbf e_x$ and $\eta_t\mathbf e_z$. Equation~\eqref{eq:jm_condition} gives effective compatibility exactly when
\[
2\sqrt2\,\eta_t\le2,
\]
so that
\begin{equation}
t_{\rm eff}=\frac{\ln 2}{2\gamma}.
\label{eq:effective_time}
\end{equation}
At every finite time, however, $\eta_t>0$ and $\mathcal D_t^\dagger$ is injective. Therefore the original Pauli pair is not exactly concealed at any finite time by Corollary~\ref{cor:invertible}. In the limit $t\to\infty$, the channel approaches the completely depolarizing map $\Delta(\rho)=I/2$, whose adjoint $\Delta^\dagger(A)=\Tr(A)I/2$ has the traceless Hermitian subspace as its kernel. Thus
\begin{equation}
t_{\rm eff}=\frac{\ln 2}{2\gamma},
\qquad t_{\rm con}=\infty.
\label{eq:time_separation}
\end{equation}
This explicitly demonstrates finite-time loss of effective incompatibility without finite-time exact operational erasure. This identifies the kernel transition, rather than merely the magnitude of decoherence, as the mechanism responsible for exact concealment of an incompatible pair. Indeed, rank loss is necessary: if an incompatible POVM pair is exactly concealed by $\mathcal E_t$, then $\mathcal E_t^\dagger$ must be noninjective.

The quantitative obstruction theorem further strengthens this: for every finite time, the approximate-concealment error has a positive lower bound controlled by the minimum gain of $\mathcal E_t^\dagger$. This remains true even for $t\ge t_{\rm eff}$, where the effective measurements are compatible.

\section{Positive Retractions and Qubit Classification}
\label{sec:retractions}

Once a channel hides observable directions, can each operational equivalence class be represented by a physically valid POVM? If a positive, unital, idempotent retraction exists, concealment reduces exactly to joint measurability of canonical representatives. We then classify the qubit channels admitting such a reduction.

\subsection{Positive-retraction principle}

\begin{theorem}[Positive-retraction principle]
\label{thm:positive_retraction}
Let $\mathcal E$ be a finite-dimensional quantum channel. Suppose there exists a real-linear map
\[
\Pi:\Herm(\mathcal H_{\rm out})
\longrightarrow
\Herm(\mathcal H_{\rm out})
\]
that is positive, unital, and idempotent, and satisfies
\[
\ker(\Pi)=\ker(\mathcal E^\dagger).
\]
Then, for every pair of finite-outcome POVMs $(M,N)$,
\begin{equation}
\begin{aligned}
&(M,N)\text{ is concealed by }\mathcal E\\
&\quad\Longleftrightarrow\quad
\bigl(\Pi(M),\Pi(N)\bigr)
\text{ is jointly measurable}.
\end{aligned}
\label{eq:positive_retraction_concealment}
\end{equation}
Its exact robustness consequence is established in Corollary~\ref{cor:positive_retraction_robustness}.
\end{theorem}

\begin{proof}
Since $\Pi^2=\Pi$,
\[
\Pi\!\left(A-\Pi(A)\right)=0
\]
for every Hermitian $A$. Hence
\[
A-\Pi(A)\in\ker(\Pi)=\ker(\mathcal E^\dagger).
\]
Positivity and unitality imply that $\Pi$ maps POVMs to POVMs and maps a joint POVM componentwise to a joint POVM.

Suppose first that $(M,N)$ is concealed. Then there exist compatible POVMs $(F,G)$ satisfying
\[
F_a-M_a\in\ker(\mathcal E^\dagger),
\qquad
G_b-N_b\in\ker(\mathcal E^\dagger).
\]
Therefore $\Pi(F_a)=\Pi(M_a)$ and $\Pi(G_b)=\Pi(N_b)$. Applying $\Pi$ componentwise to a joint POVM for $(F,G)$ proves that $(\Pi(M),\Pi(N))$ is compatible. Conversely, if $(\Pi(M),\Pi(N))$ is compatible, then it is itself a compatible representative of $(M,N)$ because $M_a-\Pi(M_a),N_b-\Pi(N_b)\in\ker(\mathcal E^\dagger)$. This proves the concealment equivalence.
\end{proof}

\begin{remark}
The theorem separates two questions. The adjoint kernel determines whether concealment is possible at all, while the existence of a positive retraction determines whether the pair-dependent feasibility problem can be reduced exactly to ordinary joint measurability. The same structure yields the robustness reduction in Corollary~\ref{cor:positive_retraction_robustness}. Characterizing the existence of positive retractions in dimensions $d\ge3$, beyond the block-decomposable class and the real-symmetric construction below, remains an open structural problem.
\end{remark}

\subsection{Classification for arbitrary qubit channels}

We now characterize exactly which qubit channels admit positive retractions, extending the unital-qubit result to arbitrary qubit channels.

A general qubit channel has the affine Bloch form
\begin{equation}
\mathcal E\!\left(\frac{I+\mathbf r\cdot\boldsymbol\sigma}{2}\right)
=
\frac{I+(T\mathbf r+\mathbf t)\cdot\boldsymbol\sigma}{2},
\label{eq:qubit_bloch}
\end{equation}
where $T$ is a real $3\times3$ matrix and $\mathbf t\in\mathbb R^3$ \cite{FujiwaraAlgoet1999,Ruskai2002}. Complete positivity imposes constraints on $T$ and $\mathbf t$, but for our classification we only need the linear structure. The adjoint acts on Hermitian operators as
\[
\mathcal E^\dagger\!\left[\frac12(\alpha I+\mathbf a\cdot\boldsymbol\sigma)\right]
=
\frac12\left[(\alpha+\mathbf t\cdot\mathbf a)I+(T^{\mathsf T}\mathbf a)\cdot\boldsymbol\sigma\right].
\]

\begin{theorem}[Positive-retraction classification for qubit channels]
\label{thm:all_qubit_retraction}
Let $\mathcal E$ be a qubit channel with affine Bloch representation $\mathbf r\mapsto T\mathbf r+\mathbf t$. There exists a positive, unital, idempotent, real-linear map
\[
\Pi:\Herm(\mathbb C^2)\longrightarrow\Herm(\mathbb C^2)
\]
satisfying
\[
\ker(\Pi)=\ker(\mathcal E^\dagger)
\]
if and only if either
\[
\rank(T)=0
\]
or
\[
\mathbf t\in\im(T).
\]

If $\rank(T)>0$ and $\mathbf t\in\im(T)$, one may choose
\[
\Pi\!\left[
\frac12(\alpha I+\mathbf a\cdot\boldsymbol\sigma)
\right]
=
\frac12
\left[
\alpha I+
(P_{\ker(T^{\mathsf T})^\perp}\mathbf a)
\cdot\boldsymbol\sigma
\right].
\]

If $\rank(T)=0$, the channel is constant: $\mathcal E(\rho)=\tau$. Hence
\[
\mathcal E^\dagger(A)=\Tr(\tau A)I
\]
and
\[
\ker(\mathcal E^\dagger)
=
\{A\in\Herm(\mathbb C^2):\Tr(\tau A)=0\}.
\]
The map
\[
\Pi(A)=\Tr(\tau A)I
\]
is positive, unital, and idempotent, and has exactly this kernel.
\end{theorem}

\begin{proof}
Every unital Hermiticity-preserving qubit map can be written in the form
\[
\Pi:(\alpha,\mathbf a)\longmapsto(\alpha+\mathbf u\cdot\mathbf a,R\mathbf a),
\]
where $R$ is a real $3\times3$ matrix and $\mathbf u\in\mathbb R^3$. Unitality fixes the identity component. Idempotence implies
\[
R^2=R,\qquad R^{\mathsf T}\mathbf u=0.
\]

If $\rank(T)>0$, then
\[
\dim\ker(\mathcal E^\dagger)=3-\rank(T)<3.
\]
Kernel equality therefore rules out $R=0$.

Choose a unit vector $\mathbf y\in\im(R^{\mathsf T})$. Since $(R^{\mathsf T})^2=R^{\mathsf T}$, we have $R^{\mathsf T}\mathbf y=\mathbf y$. Because the positive cone is self-dual under the Hilbert-Schmidt pairing, $\Pi$ is positive iff $\Pi^\dagger$ is positive. Unitality of $\Pi$ makes $\Pi^\dagger$ trace preserving. Positivity of $\Pi$ is therefore equivalent to preservation of the Bloch ball by its adjoint:
\[
|\mathbf u+R^{\mathsf T}\mathbf r|\le1
\qquad(|\mathbf r|\le1).
\]
Taking $\mathbf r=\pm\mathbf y$ gives
\[
|\mathbf u+\mathbf y|\le1,
\qquad
|\mathbf u-\mathbf y|\le1.
\]
Adding the squared inequalities yields
\[
2|\mathbf u|^2+2\le2,
\]
and hence $\mathbf u=0$.

It follows that
\[
\ker(\Pi)
=
\{(0,\mathbf a):\mathbf a\in\ker R\},
\]
whereas
\[
\ker(\mathcal E^\dagger)
=
\{(-\mathbf t\cdot\mathbf a,\mathbf a):
\mathbf a\in\ker(T^{\mathsf T})\}.
\]
Equality of these kernels implies both
\[
\ker R=\ker(T^{\mathsf T})
\]
and
\[
\mathbf t\cdot\mathbf a=0
\quad
\forall\mathbf a\in\ker(T^{\mathsf T}).
\]
The latter condition is equivalent to
\[
\mathbf t\in
\ker(T^{\mathsf T})^\perp
=
\im(T).
\]

Conversely, suppose $\rank(T)>0$ and $\mathbf t\in\im(T)$. Define $P=P_{\ker(T^{\mathsf T})^\perp}$ and
\[
\Pi(\alpha,\mathbf a)=(\alpha,P\mathbf a).
\]
Orthogonal projection is contractive, so $\Pi$ is positive; it is also unital and idempotent. Moreover,
\[
\ker(\Pi)
=
\{(0,\mathbf a):\mathbf a\in\ker(T^{\mathsf T})\}.
\]
Since $\mathbf t\in\im(T)=\ker(T^{\mathsf T})^\perp$, one has $\mathbf t\cdot\mathbf a=0$ for every $\mathbf a\in\ker(T^{\mathsf T})$. Therefore
\[
\ker(\mathcal E^\dagger)
=
\{(0,\mathbf a):\mathbf a\in\ker(T^{\mathsf T})\}
=
\ker(\Pi).
\]

The rank-zero case is handled by the constant-channel construction.
\end{proof}

\begin{corollary}[Unital qubit retraction]
Every unital qubit channel ($\mathbf t=0$) admits a positive retraction, recovering the canonical projection of Section~\ref{sec:canonical}.
\end{corollary}

\subsection{Retraction obstruction}

The classification theorem reveals that noninjectivity of the adjoint does not guarantee the existence of a positive retraction. The following example demonstrates this distinction.

\begin{example}[Obstruction example]
Consider the qubit channel with
\[
T=\diag\!\left(\frac14,\frac14,0\right),\qquad \mathbf t=\left(0,0,\frac12\right).
\]
Since $\mathbf t\notin\im(T)=\Span\{\mathbf e_1,\mathbf e_2\}$, no positive retraction exists by Theorem~\ref{thm:all_qubit_retraction}. Nevertheless, $\ker(\mathcal E^\dagger)\neq\{0\}$, so by Theorem~\ref{thm:universal_concealment}, this channel conceals some incompatible pair.

The channel is completely positive; its unnormalized Choi matrix is
\[
J_{\mathcal E}
=
\begin{pmatrix}
\frac34&0&0&\frac14\\
0&\frac14&0&0\\
0&0&\frac34&0\\
\frac14&0&0&\frac14
\end{pmatrix},
\]
with eigenvalues
\[
\frac14,\qquad
\frac34,\qquad
\frac{2-\sqrt2}{4},\qquad
\frac{2+\sqrt2}{4},
\]
all nonnegative.
\end{example}

\begin{example}[Rank-one measure-and-prepare channel]
A stronger example demonstrating nontrivial rank-deficient content is
\[
T=\diag\left(0,0,\frac12\right),\qquad \mathbf t=\left(0,0,\frac14\right).
\]
Indeed, define
\[
\tau_0=\frac12\left(I+\frac34\sigma_z\right),
\qquad
\tau_1=\frac12\left(I-\frac14\sigma_z\right).
\]
Then
\[
\mathcal E(\rho)
=
\langle0|\rho|0\rangle\,\tau_0
+
\langle1|\rho|1\rangle\,\tau_1,
\]
which is manifestly a measure-and-prepare, and hence entanglement-breaking, channel \cite{Horodecki2003}. Its affine Bloch representation is
\[
(r_x,r_y,r_z)
\longmapsto
\left(0,0,\frac12r_z+\frac14\right).
\]
This is a nonunital measure-and-prepare channel with $\mathbf t\in\im(T)=\Span\{\mathbf e_z\}$. Its adjoint is noninjective, and the positive retraction is genuine pinching onto $\Span\{I,\sigma_z\}$. This example shows that nonunital channels can admit nontrivial positive retractions, in contrast to the obstruction example above.
\end{example}

\subsection{Canonical projection for unital qubit channels}
\label{sec:canonical}

We now show that every unital qubit channel satisfies the hypotheses of Theorem~\ref{thm:positive_retraction}. The required retraction is the Euclidean orthogonal projection of Bloch vectors onto the accessible observable subspace.

Let $\mathcal E$ be a unital qubit channel with Bloch matrix $T$, and let
\[
K:=\ker(T^{\mathsf T}).
\]
Denote by $P_{K^\perp}$ the Euclidean orthogonal projection of $\mathbb R^3$ onto $K^\perp$. Define the real-linear map $\Pi_{\mathcal E}$ on Hermitian qubit operators by
\begin{equation}
\Pi_{\mathcal E}\!\left[
\frac12
\left(
\alpha I+\mathbf a\cdot\boldsymbol\sigma
\right)
\right]
:=
\frac12
\left(
\alpha I+
(P_{K^\perp}\mathbf a)\cdot\boldsymbol\sigma
\right).
\label{eq:canonical_qubit_projection}
\end{equation}
The map is extended componentwise to POVMs and POVM pairs.

\begin{lemma}[Properties of the canonical projection]
\label{lem:canonical_projection}
The map $\Pi_{\mathcal E}$ is real-linear, positive, unital, trace preserving, and idempotent. Moreover,
\begin{equation}
A-\Pi_{\mathcal E}(A)
\in\ker(\mathcal E^\dagger)
\label{eq:projection_kernel_difference}
\end{equation}
for every Hermitian qubit operator $A$, and
\begin{equation}
\ker(\Pi_{\mathcal E})
=
\ker(\mathcal E^\dagger)
=
\left\{
\frac12\,\mathbf k\cdot\boldsymbol\sigma:
\mathbf k\in K
\right\}.
\label{eq:projection_kernel_equality}
\end{equation}
Consequently,
\[
\mathcal E^\dagger(A)
=
\mathcal E^\dagger(\Pi_{\mathcal E}(A)).
\]
\end{lemma}

\begin{proof}
Write
\[
A=
\frac12
\left(
\alpha I+\mathbf a\cdot\boldsymbol\sigma
\right).
\]
If $A\succeq0$, then $\alpha\ge|\mathbf a|$. Since orthogonal projection is contractive,
\[
|P_{K^\perp}\mathbf a|
\le|\mathbf a|
\le\alpha,
\]
and therefore $\Pi_{\mathcal E}(A)\succeq0$. Thus $\Pi_{\mathcal E}$ is positive. Linearity, unitality, trace preservation, and idempotence follow directly from its definition.

Furthermore,
\[
\mathbf a-P_{K^\perp}\mathbf a\in K
=
\ker(T^{\mathsf T}).
\]
For a unital qubit channel, the adjoint acts on traceless observables according to $T^{\mathsf T}$. Hence
\[
A-\Pi_{\mathcal E}(A)
=
\frac12
\left[
(\mathbf a-P_{K^\perp}\mathbf a)
\cdot\boldsymbol\sigma
\right]
\in\ker(\mathcal E^\dagger),
\]
which proves Eq.~\eqref{eq:projection_kernel_difference}.

For Eq.~\eqref{eq:projection_kernel_equality}, if
\[
X=\frac12(\alpha I+\mathbf x\cdot\boldsymbol\sigma)
\in\ker(\mathcal E^\dagger),
\]
then
\[
0=\mathcal E^\dagger(X)
=
\frac12\left(
\alpha I+
T^{\mathsf T}\mathbf x\cdot\boldsymbol\sigma
\right).
\]
Hence $\alpha=0$ and $\mathbf x\in K$. Therefore $\Pi_{\mathcal E}(X)=0$. Conversely, if $\Pi_{\mathcal E}(X)=0$, then $\alpha=0$ and $\mathbf x\in K$, so $\mathcal E^\dagger(X)=0$. Thus $\ker(\Pi_{\mathcal E})=\ker(\mathcal E^\dagger)$.
\end{proof}

\begin{corollary}[Canonical reduction for unital qubit channels]
\label{cor:qubit_retraction}
Let $\mathcal E$ be a unital qubit channel and let $\Pi_{\mathcal E}$ be defined by Eq.~\eqref{eq:canonical_qubit_projection}. Then $\Pi_{\mathcal E}$ satisfies the hypotheses of Theorem~\ref{thm:positive_retraction}. Consequently, for every pair of finite-outcome qubit POVMs $(M,N)$,
\begin{equation}
\begin{aligned}
&(M,N)\text{ is concealed by }\mathcal E\\
&\quad\Longleftrightarrow\quad
\bigl(\Pi_{\mathcal E}(M),\Pi_{\mathcal E}(N)\bigr)
\text{ is jointly measurable}.
\end{aligned}
\label{eq:qubit_retraction_concealment}
\end{equation}
\end{corollary}

\begin{proof}
Lemma~\ref{lem:canonical_projection} establishes positivity, unitality, idempotence, and $\ker(\Pi_{\mathcal E})=\ker(\mathcal E^\dagger)$. The claim therefore follows directly from Theorem~\ref{thm:positive_retraction}.
\end{proof}

This corollary gives an explicit Bloch-space realization of the universal criterion for unital qubit channels: concealment is possible for some incompatible pair exactly when $\Pi_{\mathcal E}\neq\id$, which is equivalent to $\ker(\mathcal E^\dagger)\neq\{0\}$.

\begin{remark}
The map $\Pi_{\mathcal E}$ need not be completely positive. Complete positivity is not required here: positivity and unitality are sufficient to map POVMs and joint POVMs to POVMs and joint POVMs.
\end{remark}

\section{Higher-Dimensional Reductions}
\label{sec:higher_dim}

We now extend the positive-retraction method to higher dimensions through two structurally distinct classes: block-diagonal pinching reductions, where inaccessible coherences between subspaces are removed, and a genuinely non-block real-symmetric construction that applies in every finite dimension.

\subsection{Block-decomposition reduction}

This class includes decohering channels whose adjoints annihilate off-block coherences while remaining injective on the block-diagonal operator subspace.

Let $\{P_k\}_{k=1}^m$ be mutually orthogonal projections summing to the identity on $\mathcal H_{\rm out}$, and define the pinching map
\[
\Pi_{\mathsf P}(A)=\sum_{k=1}^m P_k A P_k.
\]

\begin{lemma}[Properties of the pinching map]
The pinching map $\Pi_{\mathsf P}$ is real-linear, positive, unital, trace preserving, and idempotent. Its kernel is
\[
\ker(\Pi_{\mathsf P}) = \{A\in\Herm(\mathcal H_{\rm out}): P_k A P_k=0\ \forall k\},
\]
i.e., the set of operators with vanishing diagonal blocks.
\end{lemma}

\begin{proof}
Each property follows directly from the definition. Positivity follows because $\Pi_{\mathsf P}(A)$ is a direct sum of positive operators when $A\succeq0$.
\end{proof}

\begin{theorem}[Block-decomposition reduction]
\label{thm:block_reduction}
Let $\mathcal E$ be a finite-dimensional quantum channel such that
\[
\ker(\mathcal E^\dagger)=\ker(\Pi_{\mathsf P})
\]
for some pinching map $\Pi_{\mathsf P}$ associated with a decomposition $\mathcal H_{\rm out}=\bigoplus_{k=1}^m\mathcal H_k$. Then a finite-outcome POVM pair $(M,N)$ is concealed by $\mathcal E$ if and only if, for every block $k$, the compressed pair
\[
\bigl(
\{P_k M_a P_k\}_a,
\{P_k N_b P_k\}_b
\bigr)
\]
is jointly measurable on $\mathcal H_k$.
\end{theorem}

\begin{proof}
By the preceding lemma, $\Pi_{\mathsf P}$ is positive, unital, and idempotent. Together with $\ker(\mathcal E^\dagger)=\ker(\Pi_{\mathsf P})$, Theorem~\ref{thm:positive_retraction} applies. Hence concealment by $\mathcal E$ is equivalent to joint measurability of $(\Pi_{\mathsf P}(M),\Pi_{\mathsf P}(N))$.

Joint measurability of the projected pair is equivalent to joint measurability in every block. To see this, suppose each block has a joint POVM $\{J^{(k)}_{ab}\}$ with marginals $\{P_k M_a P_k\}$ and $\{P_k N_b P_k\}$. Then
\[
J_{ab}=\sum_{k=1}^m J^{(k)}_{ab}
\]
is a joint POVM on the full space. Conversely, if $\{J_{ab}\}$ is a global joint POVM, then $\{P_k J_{ab}P_k\}_{a,b}$ is a joint POVM on $\mathcal H_k$, since its total sum is $P_k$, the identity operator on that block.
\end{proof}

\subsection{Non-block real-symmetric Jordan retraction}

The pinching construction produces retractions whose ranges are block-diagonal operator algebras. We now construct a positive retraction whose range is a Jordan algebra but not a complex $C^*$-subalgebra. This construction is genuinely non-block and applies in every dimension $d\ge2$.

Fix a basis and define the real-symmetrization map
\[
\Pi_{\mathbb R}(A)=\frac{A+A^{\mathsf T}}{2}.
\]
The map $\Pi_{\mathbb R}$ is the Hilbert-Schmidt orthogonal projection onto the real-symmetric subspace.

\begin{lemma}[Properties of the real-symmetrization map]
The map $\Pi_{\mathbb R}$ is real-linear, positive, unital, trace preserving, and idempotent. Its kernel is
\[
\ker\Pi_{\mathbb R}
=
\{A=A^\dagger : A^{\mathsf T}=-A\},
\]
the purely imaginary antisymmetric Hermitian subspace. Its range is the real-symmetric operator space.
\end{lemma}

\begin{proof}
Linearity, unitality, trace preservation, and idempotence are immediate. If $A\succeq0$, then for any vector $|\psi\rangle$,
\begin{align*}
\langle\psi|\Pi_{\mathbb R}(A)|\psi\rangle
&=
\frac12\langle\psi|A|\psi\rangle
+
\frac12\langle\psi|A^{\mathsf T}|\psi\rangle \\
&=
\frac12\langle\psi|A|\psi\rangle
+
\frac12\langle\psi^*|A|\psi^*\rangle
\ge0,
\end{align*}
so $\Pi_{\mathbb R}$ is positive. The kernel characterization follows directly from $A^{\mathsf T}=-A$ for Hermitian $A$.
\end{proof}

\begin{theorem}[Non-block higher-dimensional reduction]
\label{thm:real_symmetric_reduction}
For every finite dimension $d\ge2$, there exists an entanglement-breaking channel $\mathcal E$ such that
\[
\ker(\mathcal E^\dagger)=\ker\Pi_{\mathbb R}.
\]
For this channel,
\[
\begin{aligned}
(M,N)\text{ is concealed}
\quad\Longleftrightarrow\quad&
\bigl(\Pi_{\mathbb R}(M),
      \Pi_{\mathbb R}(N)\bigr)\\
&\text{is jointly measurable},
\end{aligned}
\]
and
\[
R_c(M,N\mid\mathcal E)
=
R_{\rm inc}(\Pi_{\mathbb R}M,\Pi_{\mathbb R}N).
\]
\end{theorem}

\begin{proof}
Choose a linearly independent informationally complete POVM $\{G_j\}_{j=1}^{d^2}$. Choose $d^2$ real density matrices $\{\tau_j\}_{j=1}^{d^2}$ whose real linear span equals the real-symmetric Hermitian operator space. The output states need only span the real-symmetric subspace; they are not required to be linearly independent. Such a family may be obtained as follows. Take the real symmetric pure states
\[
|i\rangle\langle i|\qquad(1\le i\le d),
\]
and
\[
\frac{(|i\rangle+|j\rangle)(\langle i|+\langle j|)}{2}
\qquad(1\le i<j\le d).
\]
These span the real-symmetric space. Since there are $d(d+1)/2$ such states, which is less than or equal to $d^2$ for $d\ge2$, repeat arbitrary members of the spanning family as needed to obtain $d^2$ outputs.

Define
\[\mathcal E(\rho)=\sum_{j=1}^{d^2}\Tr(G_j\rho)\tau_j.
\]
Then
\[
\mathcal E^\dagger(A)=\sum_j\Tr(\tau_j A)G_j.
\]
Because the $\{G_j\}$ are linearly independent,
\[
\mathcal E^\dagger(A)=0
\iff
\Tr(\tau_j A)=0\quad\forall j.
\]
By construction, the $\{\tau_j\}$ span the real-symmetric space, so
\[
\ker(\mathcal E^\dagger)
=
\left(\operatorname{span}_{\mathbb R}\{\tau_j\}\right)^\perp
=
(\operatorname{ran}\Pi_{\mathbb R})^\perp
=
\ker\Pi_{\mathbb R}.
\]
The concealment equivalence and robustness identity follow from Theorem~\ref{thm:positive_retraction} and Corollary~\ref{cor:positive_retraction_robustness}.
\end{proof}

This construction shows that the positive-retraction framework applies to genuinely noncommutative order structures in every dimension, beyond block-diagonal algebras. For $d=1$, the construction is trivial as the real-symmetric space is the full Hermitian space.

\section{Robustness Geometry, Duality, and Shift-Invariant Positive-Payoff Games}
\label{sec:quantitative}

\subsection{Concealment robustness}

We now define a measure of concealment relative to the equivalence classes induced by the adjoint kernel. The concealment robustness quantifies the minimal noise required to move a given equivalence class into the set of compatible equivalence classes. When the adjoint kernel is trivial, equivalently when the operational equivalence classes are singletons, this measure reduces to the generalized incompatibility robustness. For noninjective channels, it can depart from its unconstrained counterpart.

\begin{definition}[Concealment robustness]
\label{def:robustness}
Let $\mathcal C_{\mathcal E}$ denote the set of POVM pairs operationally concealed by the channel $\mathcal E$. The concealment robustness of a measurement pair $(\{M_a\},\{N_b\})$ under $\mathcal E$ is defined as
\begin{multline}
R_c(M,N\mid\mathcal E)
:=\inf\bigl\{t\ge0:\ \exists\ \text{POVMs }P,Q\\
\text{such that }(M^{(t)},N^{(t)})\in\mathcal C_{\mathcal E}\bigr\}.
\end{multline}
where
\[
M_a^{(t)}=\frac{M_a+tP_a}{1+t},
\qquad
N_b^{(t)}=\frac{N_b+tQ_b}{1+t}.
\]
\end{definition}

The feasible set is always nonempty for finite-outcome POVMs. Indeed, choose any $t_0>0$ satisfying $t_0\ge\max\{m-1,n-1\}$. For an $m$-outcome POVM $\{M_a\}$, define
\[
P_a=\frac{(1+t_0)I/m-M_a}{t_0}.
\]
Since $(1+t_0)/m\ge1$ and $M_a\preceq I$, we have $P_a\succeq0$. Moreover,
\[
\sum_aP_a
=
\frac{(1+t_0)I-\sum_aM_a}{t_0}
=
I,
\]
and
\[
\frac{M_a+t_0P_a}{1+t_0}=\frac{I}{m}.
\]
Applying the same construction to the $n$-outcome POVM $\{N_b\}$ makes both noisy POVMs uniform and hence jointly measurable. Therefore $R_c(M,N\mid\mathcal E)<\infty$.

To prove attainment, restrict $t$ to the compact interval $[0,t_0]$. The POVM sets with fixed finite outcome spaces are compact, and the map
\[
(t,P,Q)\longmapsto (M^{(t)},N^{(t)})
\]
is continuous. Since $\mathcal C_{\mathcal E}^{\Omega,\Lambda}$ is closed, the set
\[
\left\{
(t,P,Q):
0\le t\le t_0,\,
(M^{(t)},N^{(t)})
\in\mathcal C_{\mathcal E}^{\Omega,\Lambda}
\right\}
\]
is compact. Hence the continuous objective $t$ attains its minimum. Consequently,
\[
R_c(M,N\mid\mathcal E)=0
\iff
(M,N)\in\mathcal C_{\mathcal E}^{\Omega,\Lambda}.
\]

Throughout this section, $R_{\mathrm{inc}}$ denotes the \emph{generalized} incompatibility robustness, where the auxiliary POVMs $\{P_a\}, \{Q_b\}$ in the noise model are arbitrary~\cite{Uola2015,Designolle2019}. This should be distinguished from white-noise robustness, which corresponds to a more restrictive noise model and may yield different numerical values.

The generalized incompatibility robustness $R_{\mathrm{inc}}$ is defined analogously, requiring the noisy pair to be compatible rather than concealed, with the noise POVMs again taken to be arbitrary. For two binary qubit POVMs, the generalized incompatibility robustness can be computed via the semidefinite-programming framework of Ref.~\cite{Designolle2019}. For the orthogonal Pauli measurements, the resulting generalized incompatibility robustness is
\[
R_{\mathrm{inc}}(X,Z)=3-2\sqrt2\approx0.171573,
\]
consistent with that framework. This construction fits within the general convex-geometric framework for robustness-based quantum-resource quantification \cite{Regula2018}, with the channel-dependent concealed set $\mathcal C_{\mathcal E}$ serving as the free set. Its application to measurements, channels, and operational discrimination tasks is developed in Ref.~\cite{Takagi2019}. It therefore differs from ordinary incompatibility robustness, whose free set consists of jointly measurable POVM families themselves.

\subsection{Basic structural properties}

\begin{proposition}[Hierarchy, injective reduction, and kernel monotonicity]
\label{prop:robustness_basic}
For every finite-outcome POVM pair $(M,N)$ and channel $\mathcal E$,
\[
0\le R_c(M,N\mid\mathcal E)\le R_{\rm inc}(M,N).
\]
If $\mathcal E^\dagger$ is injective, then
\[
R_c(M,N\mid\mathcal E)=R_{\rm inc}(M,N).
\]
Moreover, if two channels with the same output space satisfy
\[
\ker(\mathcal E_1^\dagger)\subseteq\ker(\mathcal E_2^\dagger),
\]
then
\[
R_c(M,N\mid\mathcal E_2)
\le
R_c(M,N\mid\mathcal E_1).
\]
\end{proposition}

\begin{proof}
Any compatible noisy pair is concealed by choosing itself as its compatible representative, proving $R_c\le R_{\rm inc}$. If the adjoint is injective, every operational fibre is a singleton, so the concealed set coincides with the compatible set and equality follows. Finally, kernel inclusion implies $\mathcal C_{\mathcal E_1}\subseteq\mathcal C_{\mathcal E_2}$ by Theorem~\ref{thm:kernel_monotonicity}; every noise parameter feasible for $\mathcal E_1$ is therefore feasible for $\mathcal E_2$.
\end{proof}

\subsection{Retraction and block reductions of robustness}

\begin{corollary}[Robustness reduction under positive retractions]
\label{cor:positive_retraction_robustness}
Under the hypotheses of Theorem~\ref{thm:positive_retraction}, every finite-outcome POVM pair satisfies
\begin{equation}
R_c(M,N\mid\mathcal E)
=
R_{\rm inc}\!\left(\Pi(M),\Pi(N)\right).
\label{eq:positive_retraction_robustness}
\end{equation}
\end{corollary}

\begin{proof}
Let $t$ be feasible for $R_c(M,N\mid\mathcal E)$, with noise POVMs $(P,Q)$. Theorem~\ref{thm:positive_retraction} implies that
\[
\left(
\frac{\Pi(M)+t\Pi(P)}{1+t},
\frac{\Pi(N)+t\Pi(Q)}{1+t}
\right)
\]
is compatible. Hence the same $t$ is feasible for $R_{\rm inc}(\Pi(M),\Pi(N))$.

Conversely, suppose $t$ is feasible for $R_{\rm inc}(\Pi(M),\Pi(N))$, with noise POVMs $(P,Q)$. Applying $\Pi$ componentwise to a joint POVM for the noisy projected pair shows that
\[
\left(
\frac{\Pi(M)+t\Pi(P)}{1+t},
\frac{\Pi(N)+t\Pi(Q)}{1+t}
\right)
\]
is compatible. Since $\Pi(P)$ and $\Pi(Q)$ are POVMs, Theorem~\ref{thm:positive_retraction} implies that
\[
\left(
\frac{M+t\Pi(P)}{1+t},
\frac{N+t\Pi(Q)}{1+t}
\right)
\]
is concealed. The two feasible sets therefore have the same infimum.
\end{proof}

For unital qubit channels, Lemma~\ref{lem:canonical_projection} and Corollary~\ref{cor:positive_retraction_robustness} give
\begin{equation}
R_c(M,N\mid\mathcal E)
=
R_{\rm inc}\!\left(
\Pi_{\mathcal E}(M),
\Pi_{\mathcal E}(N)
\right).
\label{eq:projection_robustness_reduction}
\end{equation}
The remaining task is therefore to solve the ordinary generalized incompatibility robustness of the projected measurements. The next theorem does so exactly for orthogonal projected Bloch directions.

\begin{theorem}[Block-decomposition of robustness]
\label{thm:block_robustness}
Under the hypotheses of Theorem~\ref{thm:block_reduction},
\[
R_c(M,N\mid\mathcal E)
=
\max_{1\le k\le m}
R_{\rm inc}
\left(
\{P_k M_a P_k\}_a,
\{P_k N_b P_k\}_b
\right).
\]
\end{theorem}

\begin{proof}
Let
\[
r_k=R_{\rm inc}(M^{(k)},N^{(k)}),\qquad r=\max_k r_k.
\]
If $r=0$, every compressed pair is compatible. By Theorem~\ref{thm:block_reduction}, $(M,N)$ is concealed, and hence $R_c(M,N\mid\mathcal E)=0$. We may therefore assume $r>0$ below.

The lower bound is immediate: every globally feasible $t$ is feasible in each block, so $t\ge r_k$ for all $k$, hence $t\ge r$.

For the upper bound, note that each $r_k$ is attained by compactness of the fixed-outcome POVM sets. If block $k$ is feasible at $r_k$, write
\[
M^{(k)}+r_kP^{(k)}=(1+r_k)C_M^{(k)},
\]
with $C^{(k)}$ compatible. For the common parameter $r$, define
\[
\widetilde P^{(k)}
=
\frac{
r_kP^{(k)}+(r-r_k)C_M^{(k)}
}{r},
\]
and similarly for $Q$. Then
\[
M^{(k)}+r\widetilde P^{(k)}=(1+r)C_M^{(k)}.
\]
The direct sums
\[
\widetilde P_a=\bigoplus_k\widetilde P_a^{(k)},
\qquad
\widetilde Q_b=\bigoplus_k\widetilde Q_b^{(k)}
\]
are global noise POVMs, and the direct sum of the block joint POVMs proves global feasibility at $r$.
\end{proof}

\subsection{Exact qubit robustness geometry}

The projection theorem yields a powerful exact formula for binary qubit POVMs whose projected Bloch vectors are orthogonal.

Define the function
\[
t_*(a,b):=
a+b+1-\sqrt{2(1+a)(1+b)}.
\]
For $a,b\in[0,1]$ with $a^2+b^2>1$, this is positive and satisfies $t_*(a,b)\le\min(a,b)$. Indeed,
\begin{align*}
t_*(a,b)>0
&\iff (a+b+1)^2>2(1+a)(1+b)\\
&\iff a^2+b^2>1,
\end{align*}
while $t_*(a,b)\le a$ and $t_*(a,b)\le b$ follow from $b+1\le2(1+a)$ and $a+1\le2(1+b)$, respectively.

\begin{theorem}[Exact robustness for orthogonal projected directions]
\label{thm:orthogonal_projected_robustness}
Let $\mathcal E$ be a unital qubit channel. Consider two unbiased binary POVMs whose canonically projected Bloch vectors are
\[
P_{K^\perp}\mathbf m=a\,\mathbf u,
\qquad
P_{K^\perp}\mathbf n=b\,\mathbf v,
\]
where
\[
0\le a,b\le1,
\qquad
\mathbf u\cdot\mathbf v=0.
\]
Then
\begin{equation}
R_c(M,N\mid\mathcal E)
=
\begin{cases}
0,
& a^2+b^2\le1,\\[1ex]
t_*(a,b),
& a^2+b^2>1.
\end{cases}
\label{eq:orthogonal_projected_robustness}
\end{equation}
\end{theorem}

\begin{proof}
By Eq.~\eqref{eq:projection_robustness_reduction}, it is sufficient to compute the generalized incompatibility robustness of the projected pair. After a unitary rotation, the projected difference observables may be written as
\[
\widehat M=a\sigma_x,
\qquad
\widehat N=b\sigma_y.
\]

If $a^2+b^2\le1$, the projected pair is compatible by the unbiased binary-qubit joint-measurability criterion, and the robustness vanishes.

Assume now that $a^2+b^2>1$. The incompatibility assumption implies $t_*(a,b)>0$. Moreover, $t_*(a,b)\le a$ and $t_*(a,b)\le b$, so
\[
\alpha:=a-t_*(a,b)\ge0,
\qquad
\beta:=b-t_*(a,b)\ge0.
\]

For the upper bound, choose sharp noise POVMs with difference observables $-\sigma_x$ and $-\sigma_y$. The noisy Bloch vectors are
\[
\frac{a-t}{1+t}\mathbf e_x,
\qquad
\frac{b-t}{1+t}\mathbf e_y.
\]
They are compatible precisely when
\[
(a-t)^2+(b-t)^2\le(1+t)^2.
\]
The smallest nonnegative solution is exactly $t=t_*(a,b)$, proving
\[
R_c(M,N\mid\mathcal E)\le t_*(a,b).
\]

For the lower bound, we first note the following compatibility inequality. Let $(F,G)$ be arbitrary compatible binary qubit POVMs, and write the Bloch-vector parts of their difference observables as $\mathbf f$ and $\mathbf g$. For all $\alpha,\beta\ge0$,
\begin{equation}
\alpha f_x+\beta g_y
\le
\sqrt{\alpha^2+\beta^2}.
\label{eq:orthogonal_compatibility_witness}
\end{equation}
Indeed, if $\{J_{su}\}_{s,u=\pm1}$ is a joint POVM, then
\[
\widehat F=\sum_{s,u}sJ_{su},
\qquad
\widehat G=\sum_{s,u}uJ_{su},
\]
and therefore
\begin{align*}
\alpha f_x+\beta g_y
&=
\frac12
\sum_{s,u}
\Tr\!\left[
J_{su}
(\alpha s\sigma_x+\beta u\sigma_y)
\right]
\\
&\le
\frac12
\sum_{s,u}
\Tr(J_{su})
\left\|
\alpha s\sigma_x+\beta u\sigma_y
\right\|_\infty
\\
&=
\sqrt{\alpha^2+\beta^2},
\end{align*}
because
\[
\left\|
\alpha s\sigma_x+\beta u\sigma_y
\right\|_\infty
=
\sqrt{\alpha^2+\beta^2}
\]
and $\sum_{s,u}J_{su}=I$.

Suppose that a noise parameter $t$ is feasible. Let $p_x$ and $q_y$ be the corresponding components of the two noise difference observables. Since they arise from binary POVMs,
\[
p_x\ge-1,
\qquad
q_y\ge-1.
\]
Applying Eq.~\eqref{eq:orthogonal_compatibility_witness} to the noisy compatible pair gives
\[
\alpha
\frac{a+tp_x}{1+t}
+
\beta
\frac{b+tq_y}{1+t}
\le
\sqrt{\alpha^2+\beta^2}.
\]
Hence
\begin{equation}
\alpha a+\beta b-t(\alpha+\beta)
\le
(1+t)\sqrt{\alpha^2+\beta^2}.
\label{eq:orthogonal_lower_intermediate}
\end{equation}

By the defining equation for $t_*(a,b)$,
\[
(a-t_*)^2+(b-t_*)^2=(1+t_*)^2,
\]
and therefore
\[
\sqrt{\alpha^2+\beta^2}=1+t_*.
\]
Furthermore,
\begin{align*}
\alpha a+\beta b
&=
\alpha(\alpha+t_*)+\beta(\beta+t_*)
\\
&=
(1+t_*)^2+t_*(\alpha+\beta).
\end{align*}
Substituting these identities into Eq.~\eqref{eq:orthogonal_lower_intermediate} yields
\[
(t_*-t)
\left(
1+t_*+\alpha+\beta
\right)
\le0.
\]
The factor in parentheses is strictly positive, so every feasible $t$ satisfies $t\ge t_*$. Thus
\[
R_c(M,N\mid\mathcal E)\ge t_*(a,b).
\]
Together with the upper bound, this proves Eq.~\eqref{eq:orthogonal_projected_robustness}.
\end{proof}

\begin{corollary}[Bloch-rank-two family]
\label{cor:exact_intermediate_robustness}
Let
\[
T_{\mathcal E}
=
\diag\!\left(\frac12,\frac12,0\right).
\]
This Bloch matrix corresponds to the CPTP Pauli channel
\[
\mathcal E(\rho)
=
\frac12\rho
+\frac14\sigma_x\rho\sigma_x
+\frac14\sigma_y\rho\sigma_y.
\]
Consider the unbiased binary qubit POVMs with Bloch vectors
\[
\mathbf m_r=(r,0,\sqrt{1-r^2}),
\qquad
\mathbf n=(0,1,0).
\]
Then
\[
R_c(M,N\mid\mathcal E)
=
(\sqrt{1+r}-1)^2,
\qquad 0\le r\le1.
\]
The gap is maximal at $r=0$, where
\[
R_{\rm inc}-R_c=3-2\sqrt2,
\]
and decreases continuously to zero at $r=1$.
\end{corollary}

\begin{proof}
The canonical projections have orthogonal Bloch vectors of lengths $a=r$ and $b=1$, so the result follows from Theorem~\ref{thm:orthogonal_projected_robustness}.
\end{proof}

\begin{figure}[!t]
\centering
\begin{tikzpicture}[x=4.2cm,y=20cm]

\draw[->] (-0.03,0) -- (1.07,0) node[right] {$r$};
\draw[->] (0,-0.008) -- (0,0.195) node[above] {$R$};

\draw (0,0) node[below left] {$0$};
\draw (1,0) node[below] {$1$};

\draw[red,dashed]
(0,0.17157287525381)
--
(1,0.17157287525381)
node[right] {$R_{\rm inc}$};

\draw[blue,thick]
plot[
domain=0:1,
samples=120,
smooth
]
(\x,{(\x+2-2*sqrt(\x+1))})
node[below right] {$R_c$};

\node[align=center] at (0.48,0.115)
{strict separation for\\$0\le r<1$};

\node[align=center] at (0.42,0.070)
{$\displaystyle
\sup_{0\le r\le1}(R_{\rm inc}-R_c)
=3-2\sqrt2$};

\end{tikzpicture}
\caption{
(Color online) Exact concealment robustness
$R_c=(\sqrt{1+r}-1)^2$ (blue curve) compared with generalized
incompatibility robustness
$R_{\rm inc}=3-2\sqrt2$ (red dashed line).
The gap is strictly positive for $0\le r<1$, vanishes at $r=1$,
and is maximal at $r=0$ with value $3-2\sqrt2$.
}
\label{fig:robustness_comparison}
\end{figure}
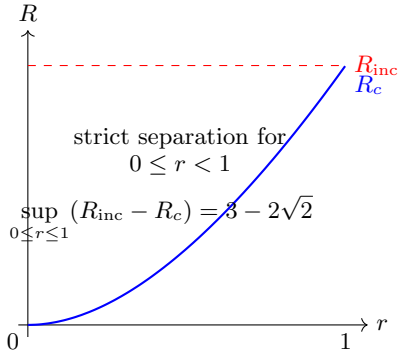

Fig.~\ref{fig:robustness_comparison} shows the continuous separation that is strict for \(0\le r<1\) and vanishes at \(r=1\).

\subsection{Quotient duality}

To make concealment robustness experimentally testable, we derive a quotient-space dual whose witnesses are evaluable from channel-output statistics and yield an additive advantage in positive-payoff games.

\begin{proposition}[Operational invariance of concealment robustness]
\label{prop:robustness_quotient_invariance}
Let $\mathbf M=\{M_{a|x}\}_{a,x}$ and $\mathbf M'=\{M'_{a|x}\}_{a,x}$ be two finite POVM pairs with the same outcome sets. If
\[
M_{a|x}-M'_{a|x}\in\ker(\mathcal E^\dagger)
\qquad
\forall a,x,
\]
then
\[
R_c(\mathbf M\mid\mathcal E)
=
R_c(\mathbf M'\mid\mathcal E).
\]
Consequently, $R_c$ depends only on the operational equivalence classes of the measurement effects.
\end{proposition}

\begin{proof}
Let $t\ge0$ and let $\mathbf P=\{P_{a|x}\}_{a,x}$ be any noise pair. The corresponding noisy effects satisfy
\begin{align*}
&\frac{M_{a|x}+tP_{a|x}}{1+t}
-
\frac{M'_{a|x}+tP_{a|x}}{1+t}\\
&\qquad=
\frac{M_{a|x}-M'_{a|x}}{1+t}
\in\ker(\mathcal E^\dagger).
\end{align*}
Hence the two noisy pairs belong componentwise to the same operational equivalence classes. One noisy pair is concealed if and only if the other is concealed. Therefore every value of $t$ feasible for $\mathbf M$ is feasible for $\mathbf M'$, and conversely. Taking the infimum proves the claim.
\end{proof}

Let
\[
d_{\rm in}:=\dim\mathcal H_{\rm in}.
\]
Let
\[
\mathsf V
=
\bigoplus_{x=0}^{1}
\bigoplus_{a\in\Omega_x}
\Herm(\mathcal H_{\rm out}),
\qquad
\mathsf K_{\mathcal E}
=
\bigoplus_{x,a}\ker(\mathcal E^\dagger),
\]
and let $\pi:\mathsf V\to\mathsf V/\mathsf K_{\mathcal E}$ be the componentwise quotient map.

Let $\mathsf P$ denote the convex set of all POVM pairs with the specified outcomes, and let $\mathsf C_{\mathcal E}\subseteq\mathsf P$ denote the concealed pairs.

This construction is related to linear incompatibility witnesses \cite{Carmeli2019} and to dual semidefinite-programming formulations of incompatibility robustness \cite{Designolle2019}, but the witnesses here act on operational equivalence classes and must annihilate the adjoint kernel.

\begin{theorem}[Quotient dual of concealment robustness]
\label{thm:quotient_dual}
Let
\[
\overline{\mathsf P}:=\pi(\mathsf P),
\qquad
\overline{\mathsf C}_{\mathcal E}
:=
\pi(\mathsf C_{\mathcal E}).
\]
Then
\begin{equation}
\label{eq:quotient_dual}
1+R_c(\mathbf M\mid\mathcal E)
=
\max_{\omega}\;
\omega\!\left(\pi(\mathbf M)\right),
\end{equation}
where the maximum is over real linear functionals
$\omega\in(\mathsf V/\mathsf K_{\mathcal E})^*$ satisfying
\begin{align}
\omega(\overline{\mathbf P})&\ge0
&&
\forall\overline{\mathbf P}\in\overline{\mathsf P},
\label{eq:dual_noise_positive}\\
\omega(\overline{\mathbf F})&\le1
&&
\forall\overline{\mathbf F}
\in\overline{\mathsf C}_{\mathcal E}.
\label{eq:dual_free_bound}
\end{align}
The maximum is attained and there is no duality gap.
\end{theorem}

\begin{proof}
Let
\[
\mathsf Q:=\mathsf V/\mathsf K_{\mathcal E}.
\]
Define the cones
\[
\mathsf K_P
:=
\{\lambda\overline{\mathbf P}:
\lambda\ge0,\,
\overline{\mathbf P}\in\overline{\mathsf P}\},
\]
and
\[
\mathsf K_C
:=
\{\lambda\overline{\mathbf F}:
\lambda\ge0,\,
\overline{\mathbf F}\in
\overline{\mathsf C}_{\mathcal E}\}.
\]

There exists a linear normalization functional $u\in\mathsf Q^*$ satisfying
\[
u(\overline{\mathbf P})=
u(\overline{\mathbf F})=1
\]
for all normalized POVM pairs. Hence both compact bases lie in the affine hyperplane
\[
\{X\in\mathsf Q:u(X)=1\}.
\]
Since both bases are compact and contained in this affine
hyperplane, neither contains the origin. Consequently, their
conic hulls $\mathsf K_P$ and $\mathsf K_C$ are closed convex cones.

For two settings, one convenient choice for $u$ is
\[
u(\pi(\mathbf X))
=
\frac{1}{2d_{\rm in}}
\sum_{x=0}^{1}
\Tr\!\left[
\mathcal E^\dagger
\left(
\sum_{a\in\Omega_x}X_{a|x}
\right)
\right].
\]
This is well defined on the quotient because $\mathcal E^\dagger$ annihilates $\mathsf K_{\mathcal E}$.

Writing $\overline{\mathbf M}:=\pi(\mathbf M)$, the robustness takes the conic form
\[
1+R_c(\mathbf M\mid\mathcal E)
=
\min_{Y,Z}\;u(Y)
\]
subject to
\[
Y-Z=\overline{\mathbf M},
\qquad
Y\in\mathsf K_C,
\qquad
Z\in\mathsf K_P.
\]
Indeed, $Y=s\overline{\mathbf F}$ and $Z=(s-1)\overline{\mathbf P}$ reproduce exactly the defining robustness relation.

Introducing a dual functional $\omega\in\mathsf Q^*$, the Lagrangian is
\[
u(Y)+\omega(\overline{\mathbf M}-Y+Z).
\]
The infimum over $Y\in\mathsf K_C$ and $Z\in\mathsf K_P$ is finite precisely when
\[
u-\omega\in\mathsf K_C^*,
\qquad
\omega\in\mathsf K_P^*.
\]
On the normalized bases of the two cones, these conditions become
\[
\omega(\overline{\mathbf F})\le1
\quad
\forall\overline{\mathbf F}
\in\overline{\mathsf C}_{\mathcal E},
\]
and
\[
\omega(\overline{\mathbf P})\ge0
\quad
\forall\overline{\mathbf P}
\in\overline{\mathsf P}.
\]
This yields Eq.~\eqref{eq:quotient_dual}.

For strong duality, let $\mathbf U$ be the uniform POVM pair and choose the strictly positive joint POVM
\[
J_{ab}=\frac{I}{|\Omega_0||\Omega_1|}.
\]
Let $\mathsf A=\operatorname{aff}(\mathsf P)$ be the affine space of normalized Hermitian POVM pairs. The marginal map from Hermitian joint arrays to $\mathsf A$ is surjective on the corresponding tangent spaces: for arbitrary perturbations satisfying
\[
\sum_a\Delta F_a=0,
\qquad
\sum_b\Delta G_b=0,
\]
one may take
\[
\Delta J_{ab}
=
\frac{\Delta F_a}{|\Omega_1|}
+
\frac{\Delta G_b}{|\Omega_0|}.
\]
Because every uniform joint effect is positive definite, all sufficiently small perturbations of this form preserve positivity. Hence the compatible set contains a relative neighbourhood of $\mathbf U$ in the full affine space $\mathsf A$. In particular,
\[
\mathbf U\in\operatorname{ri}(\mathsf C_{\mathcal E}),
\qquad
\operatorname{aff}(\mathsf C_{\mathcal E})=\mathsf A.
\]
A linear map sends relative interiors onto the relative interiors of its images, so
\[
\pi(\mathbf U)
\in
\operatorname{ri}(\overline{\mathsf C}_{\mathcal E})
\cap
\operatorname{ri}(\overline{\mathsf P}).
\]
Now choose $s>\max_x|\Omega_x|$ sufficiently large and define the normalized noise pair
\[
\mathbf P_s
=
\frac{s\mathbf U-\mathbf M}{s-1}.
\]
Every effect of $\mathbf P_s$ is positive definite, so $\pi(\mathbf P_s)\in\operatorname{ri}(\overline{\mathsf P})$. Setting
\[
Y=s\pi(\mathbf U),
\qquad
Z=(s-1)\pi(\mathbf P_s),
\]
gives $Y-Z=\overline{\mathbf M}$ with
\[
Y\in\operatorname{ri}(\mathsf K_C),
\qquad
Z\in\operatorname{ri}(\mathsf K_P)
\]
in the natural linear spans of the cones. Thus the equality-constrained conic program satisfies the relative Slater condition. Finite-dimensional conic duality \cite{BoydVandenberghe2004} gives zero duality gap and dual attainment. Primal attainment follows from the compactness argument given after Definition~\ref{def:robustness}.
\end{proof}

\begin{corollary}[Concealment witnesses]
\label{cor:concealment_witness}
A POVM pair $\mathbf M$ is not concealed by $\mathcal E$ if and only if there exists a quotient functional $\omega$ satisfying Eqs.~\eqref{eq:dual_noise_positive} and \eqref{eq:dual_free_bound} such that
\[
\omega(\pi(\mathbf M))>1.
\]
Moreover, an optimal witness satisfies
\[
\omega(\pi(\mathbf M))
=
1+R_c(\mathbf M\mid\mathcal E).
\]
\end{corollary}

\begin{theorem}[Restricted-access witness representation]
\label{thm:accessible_witness}
Every quotient witness in Theorem~\ref{thm:quotient_dual} can be evaluated from channel-output statistics.

More precisely, there exist Hermitian operators $W_{a|x}\in(\ker\mathcal E^\dagger)^\perp=\im\mathcal E$ such that
\[
\omega(\pi(\mathbf M))
=
\sum_{x,a}\Tr(W_{a|x}M_{a|x}).
\]
For each $a,x$, one may write $W_{a|x}=\mathcal E(H_{a|x})$ for some Hermitian input operator $H_{a|x}$. If $\{\rho_j\}_j$ is a tomographically complete family, then
\[
H_{a|x}
=
\sum_j c_{a|x,j}\rho_j
\]
with real coefficients $c_{a|x,j}$, and consequently
\begin{equation}
\label{eq:accessible_witness_statistics}
\omega(\pi(\mathbf M))
=
\sum_{x,a,j}
c_{a|x,j}
\Tr\!\left[M_{a|x}\mathcal E(\rho_j)\right].
\end{equation}
Thus the optimal concealment witness is a real linear combination of experimentally accessible channel-output probabilities.
\end{theorem}

\begin{proof}
The dual of a quotient space is naturally identified with the annihilator of the quotient subspace:
\[
(\mathsf V/\mathsf K_{\mathcal E})^*
\cong
\mathsf K_{\mathcal E}^{\perp}.
\]
Hence each component of the lifted witness belongs to $(\ker\mathcal E^\dagger)^\perp$. Finite-dimensional duality gives
\[
(\ker\mathcal E^\dagger)^\perp=\im\mathcal E.
\]
Therefore $W_{a|x}=\mathcal E(H_{a|x})$ for some Hermitian $H_{a|x}$. Since a tomographically complete family spans the Hermitian input space, each $H_{a|x}$ is a real linear combination of the states $\rho_j$. Substitution gives Eq.~\eqref{eq:accessible_witness_statistics}.
\end{proof}

\subsection{Shift-invariant normalised positive-payoff operational games}

We now give an exact operational interpretation of the concealment robustness using shift-invariant normalised positive-payoff games. This interpretation follows the operational resource-theoretic framework of \cite{Uola2019,Skrzypczyk2019,Takagi2019}.

Fix a tomographically complete input family $\{\rho_j\}_{j=1}^{d_{\rm in}^2}$. A \emph{positive-payoff channel-output game} is specified by nonnegative coefficients $c_{a|x,j}\ge0$, with score
\begin{equation}
S_{\mathcal G}(\mathbf A)
:=
\sum_{a,x,j}
c_{a|x,j}
\Tr\!\left[
A_{a|x}\mathcal E(\rho_j)
\right].
\label{eq:positive_channel_output_game_score}
\end{equation}
Define the minimum score over all POVM pairs and the maximum over concealed pairs:
\[
\alpha_{\mathcal G}
:=
\min_{\mathbf P\in\mathsf P}
S_{\mathcal G}(\mathbf P),
\qquad
\beta_{\mathcal G}^{\rm con}
:=
\max_{\mathbf F\in\mathsf C_{\mathcal E}}
S_{\mathcal G}(\mathbf F).
\]
The game is \emph{normalised} when
\[
\beta_{\mathcal G}^{\rm con}-\alpha_{\mathcal G}\le1.
\]

\begin{theorem}[Shift-invariant positive-payoff operational advantage]
\label{thm:shift_invariant_positive_operational_advantage}
For every finite-outcome POVM pair $\mathbf M$,
\begin{equation}
R_c(\mathbf M\mid\mathcal E)
=
\max_{\mathcal G}
\left[
S_{\mathcal G}(\mathbf M)
-
\beta_{\mathcal G}^{\rm con}
\right],
\label{eq:shift_invariant_positive_operational_advantage}
\end{equation}
where the maximum is over all normalised positive-payoff channel-output games. The maximum is attained.

Thus $R_c(\mathbf M\mid\mathcal E)$ is exactly the largest additive score advantage of $\mathbf M$ over all measurement pairs concealed by $\mathcal E$, using only channel-output statistics with nonnegative payoffs and the shift-invariant normalisation.
\end{theorem}

\begin{proof}
Let $\omega_*$ be an optimal quotient functional in Theorem~\ref{thm:quotient_dual}. By Theorem~\ref{thm:accessible_witness}, there exist real coefficients $d_{a|x,j}$ (possibly signed) such that
\[
\omega_*(\pi(\mathbf A))
=
\sum_{a,x,j} d_{a|x,j}\Tr[A_{a|x}\mathcal E(\rho_j)].
\]

For each $(x,j)$, choose $\lambda_{x,j}\ge -\min_a d_{a|x,j}$ and define
\[
c_{a|x,j}=d_{a|x,j}+\lambda_{x,j}\ge0.
\]
Since $\sum_a\Tr[A_{a|x}\mathcal E(\rho_j)]=1$ for every POVM $\mathbf A$,
\[
S_{\mathcal G}(\mathbf A)
=
\omega_*(\pi(\mathbf A))
+
C,
\qquad
C:=\sum_{x,j}\lambda_{x,j}.
\]

Dual feasibility and positivity on $\mathsf P$ give
\[
0\le c_*\le1,
\]
where\[
c_*:=
\max_{\mathbf F\in\mathsf C_{\mathcal E}}
\omega_*(\pi(\mathbf F)).
\]
Optimality forces $c_*=1$. Indeed, if $0<c_*<1$,
then $\omega_*/c_*$ would remain dual feasible and would have a
strictly larger objective value. If $c_*=0$, arbitrary positive
rescaling would remain feasible while increasing the positive
objective $\omega_*(\pi(\mathbf M))=1+R_c$, contradicting finiteness
of the optimum.

Consequently,
\[
\beta_{\mathcal G}^{\rm con}=1+C,
\]
and therefore
\[
S_{\mathcal G}(\mathbf M)
-\beta_{\mathcal G}^{\rm con}
=
\omega_*(\pi(\mathbf M))-1
=
R_c(\mathbf M\mid\mathcal E).
\]
Moreover,
\[
\beta_{\mathcal G}^{\rm con}-\alpha_{\mathcal G}
=
1-\min_{\mathbf P\in\mathsf P}
\omega_*(\pi(\mathbf P))
\le1,
\]
so the constructed game is normalised.

Conversely, let $\mathcal G$ be any normalised positive-payoff game
and let $t$ be feasible in the definition of
$R_c(\mathbf M\mid\mathcal E)$. Then there exist a POVM pair
$\mathbf P$ and a concealed pair $\mathbf F$ such that
\[
\pi(\mathbf M)+t\pi(\mathbf P)
=
(1+t)\pi(\mathbf F).
\]
Because $S_{\mathcal G}$ depends only on channel-output
probabilities, it factors through the quotient, and hence
\[
S_{\mathcal G}(\mathbf M)
+tS_{\mathcal G}(\mathbf P)
=
(1+t)S_{\mathcal G}(\mathbf F).
\]
It follows that
\begin{align*}
S_{\mathcal G}(\mathbf M)
-\beta_{\mathcal G}^{\rm con}
&\le
t\left(
\beta_{\mathcal G}^{\rm con}
-\alpha_{\mathcal G}
\right)\\
&\le t.
\end{align*}
Taking the infimum over feasible $t$ gives
\[
S_{\mathcal G}(\mathbf M)
-\beta_{\mathcal G}^{\rm con}
\le
R_c(\mathbf M\mid\mathcal E).
\]
Together with the game constructed from $\omega_*$, this proves
the equality and attainment.
\end{proof}

This removes the need for signed coefficients while preserving the exact operational interpretation through the shift-invariant normalisation. This is in the spirit of operational resource-theoretic quantification \cite{Regula2018,Takagi2019,Skrzypczyk2019}.

\paragraph*{Operational protocol.}
The certificate requires only three steps: prepare a tomographically complete collection of input states, pass each state through the restricted channel, and estimate the outcome probabilities entering the witness score. A score exceeding the compatible bound, by more than the finite-sample confidence radius, certifies that no compatible representatives can reproduce the observed channel-output statistics. Neither reconstruction of the original POVM operators nor access to states outside the channel image is required.

\section{Ancilla stability}
\label{sec:ancilla}

Operational concealment of the lifted measurement pair is invariant under adjoining an arbitrary passive reference system, even when the compatible representatives on the enlarged output space are allowed to be arbitrary joint POVMs.

Throughout this subsection, concealment for $\mathcal E_R=\mathcal E\otimes\id_R$ is understood relative to a tomographically complete input family on $\mathcal H_{\rm in}\otimes\mathcal H_R$, equivalently through the adjoint characterization of Theorem~\ref{thm:necessary_sufficient}.

\begin{theorem}[Passive ancilla invariance for concealment]
\label{thm:ancilla_stability}
Let $\mathcal E$ be a finite-dimensional channel, and let $R$ be a finite-dimensional reference system. Define $\mathcal E_R=\mathcal E\otimes\id_R$ and $M^R_a=M_a\otimes I_R$, $N^R_b=N_b\otimes I_R$.

Then:
\begin{enumerate}
\item $(M,N)$ is concealed by $\mathcal E$ if and only if $(M^R,N^R)$ is concealed by $\mathcal E_R$.
\item $R_c(M^R,N^R\mid\mathcal E_R)=R_c(M,N\mid\mathcal E)$.
\item $\epsilon_c(M^R,N^R\mid\mathcal E_R)=\epsilon_c(M,N\mid\mathcal E)$.
\end{enumerate}
\end{theorem}

\begin{proof}
We prove the equivalence by constructing reduction maps.

For a fixed reference state $\tau\in\mathcal S(\mathcal H_R)$, define the completely positive, unital map
\[
\Gamma_\tau^S(X)=\operatorname{Tr}_R\!\left[
(I_S\otimes\tau_R^{1/2})X(I_S\otimes\tau_R^{1/2})
\right]
\]
for each system $S\in\{\text{in},\text{out}\}$. Equivalently, $\Gamma_\tau^S$ is the Heisenberg adjoint of the channel $\rho\mapsto\rho\otimes\tau$.

For every system operator $A$,
\[
\Gamma_\tau^S(A\otimes I_R)=A.
\]

Moreover, writing $\mathcal E(\rho)=\sum_jK_j\rho K_j^\dagger$, we have
\[
\Gamma_\tau^{\rm in}
\!\left[(\mathcal E^\dagger\otimes\id_R)(X)\right]
=
\sum_jK_j^\dagger\Gamma_\tau^{\rm out}(X)K_j
=
\mathcal E^\dagger\!\left(\Gamma_\tau^{\rm out}(X)\right).
\]
Thus
\[
\mathcal E^\dagger\circ\Gamma_\tau^{\rm out}
=
\Gamma_\tau^{\rm in}\circ
(\mathcal E^\dagger\otimes\id_R).
\]

Now suppose $(M^R,N^R)$ is concealed by $\mathcal E_R$. Then there exist compatible representatives $(F^R,G^R)$ such that
\[
(\mathcal E^\dagger\otimes\id_R)(F^R_a-M^R_a)=0,
\qquad
(\mathcal E^\dagger\otimes\id_R)(G^R_b-N^R_b)=0.
\]
Applying $\Gamma_\tau^{\rm in}$ to both sides and using the intertwining identity gives
\[
\mathcal E^\dagger\!\left(\Gamma_\tau^{\rm out}(F^R_a)-M_a\right)=0,
\]
and similarly for $G$. Since $\Gamma_\tau^S$ is completely positive and unital, $\Gamma_\tau^{\rm out}(F^R_a)$ and $\Gamma_\tau^{\rm out}(G^R_b)$ are POVMs on the system. Moreover, $\Gamma_\tau$ is compatibility preserving: if $\{J_{ab}^{R}\}$ is a joint POVM for $(F^R,G^R)$, then $\{\Gamma_\tau^{\rm out}(J_{ab}^{R})\}$ is a joint POVM for $(\Gamma_\tau^{\rm out}(F^R_a),\Gamma_\tau^{\rm out}(G^R_b))$. Hence $(M,N)$ is concealed by $\mathcal E$.

The converse direction is immediate: if $(F,G)$ conceals $(M,N)$ on the system, then $(F\otimes I_R,G\otimes I_R)$ conceals $(M^R,N^R)$ on the enlarged system.

For the robustness equality, first let $t$ be feasible for $R_c(M,N\mid\mathcal E)$, with noise POVMs $(P,Q)$. Then $(P\otimes I_R,Q\otimes I_R)$ is a valid pair of noise POVMs on the enlarged system, and the corresponding noisy lifted pair is concealed by $\mathcal E_R$. Hence
\[
R_c(M^R,N^R\mid\mathcal E_R)
\le
R_c(M,N\mid\mathcal E).
\]

Conversely, suppose that $t$ is feasible for $R_c(M^R,N^R\mid\mathcal E_R)$ with arbitrary noise POVMs on the enlarged system, $(P^R,Q^R)$. Define
\[
P_a:=\Gamma_\tau^{\rm out}(P_a^R),
\qquad
Q_b:=\Gamma_\tau^{\rm out}(Q_b^R).
\]
Since $\Gamma_\tau^{\rm out}$ is positive and unital, $(P,Q)$ is a pair of POVMs. Moreover,
\[
\Gamma_\tau^{\rm out}
\left(
\frac{M_a^R+tP_a^R}{1+t}
\right)
=
\frac{M_a+tP_a}{1+t},
\]
and similarly for $N$. The reduction argument above shows that this system noisy pair is concealed by $\mathcal E$. Therefore
\[
R_c(M,N\mid\mathcal E)
\le
R_c(M^R,N^R\mid\mathcal E_R).
\]
Combining the two inequalities proves the claimed equality.

For the approximate-concealment error, let $(F^R,G^R)$ be any compatible pair on the enlarged output space and define
\[
F_a:=\Gamma_\tau^{\rm out}(F_a^R),
\qquad
G_b:=\Gamma_\tau^{\rm out}(G_b^R).
\]
This is a compatible system pair. By the intertwining identity and operator-norm contractivity of the positive unital map $\Gamma_\tau^{\rm in}$,
\begin{align*}
\left\|
\mathcal E^\dagger(F_a-M_a)
\right\|_\infty
&=
\left\|
\Gamma_\tau^{\rm in}
\left[
(\mathcal E^\dagger\otimes\id_R)
(F_a^R-M_a^R)
\right]
\right\|_\infty\\
&\le
\left\|
(\mathcal E^\dagger\otimes\id_R)
(F_a^R-M_a^R)
\right\|_\infty .
\end{align*}
The same inequality holds for $G_b-N_b$. Taking the infimum over enlarged compatible pairs gives
\[
\epsilon_c(M,N\mid\mathcal E)
\le
\epsilon_c(M^R,N^R\mid\mathcal E_R).
\]

Conversely, tensoring any compatible system pair $(F,G)$ with $I_R$ gives a compatible enlarged pair, and
\[
\left\|
(\mathcal E^\dagger\otimes\id_R)
[(F_a-M_a)\otimes I_R]
\right\|_\infty
=
\left\|
\mathcal E^\dagger(F_a-M_a)
\right\|_\infty .
\]
Hence the reverse inequality holds, proving
\[
\epsilon_c(M^R,N^R\mid\mathcal E_R)
=
\epsilon_c(M,N\mid\mathcal E).
\]
\end{proof}

\begin{corollary}[Passive ancillas cannot activate concealment]
If a POVM pair is compatible on the original system, it remains compatible after tensoring with the identity on the reference. Conversely, if $(M\otimes I_R,N\otimes I_R)$ is concealed by $\mathcal E\otimes\id_R$, then $(M,N)$ is concealed by $\mathcal E$.
\end{corollary}

\begin{remark}
The theorem concerns the lifted target measurements $M_a\otimes I_R$ and $N_b\otimes I_R$ under a single use of $\mathcal E\otimes\id_R$. The compatible representatives and noise POVMs on the enlarged output space may be arbitrary joint system--reference POVMs. The result does not cover arbitrary joint target measurements, adaptive channel use, or genuinely multi-use strategies.
\end{remark}

\section{Discussion}
\label{sec:discussion}

Under tomographically complete input access, exact concealment is equivalent to compatibility relative to the channel-accessible state set $\mathcal S_{\mathcal E}$, and therefore belongs to the broader restricted-state framework \cite{Heinosaari2021FewStates}. The channel-generated setting supplies additional structure: the annihilator of the accessible state span is exactly $\ker(\mathcal E^\dagger)$, which turns a restricted-testing problem into an affine-fibre problem and yields the universal existence criterion and kernel preorder. This structural simplification is the main advantage over general restricted-state compatibility. Related work on incompatibility under restricted subspace access \cite{Uola2021Subspaces} addresses a different but complementary scenario.

The central physical distinction is between rank loss and contraction. Rank loss---not contraction---is the mechanism underlying exact concealment of otherwise incompatible measurements. Contraction can make the effective POVMs jointly measurable while the adjoint remains injective; depolarizing and phase-damping dynamics exhibit this at finite time. Exact concealment of an incompatible pair requires a new hidden observable direction and therefore a loss of linear rank. The finite-time no-go theorem shows that this qualitative transition cannot occur under regular finite-dimensional time-local evolution. The quantitative obstruction theorem strengthens this by showing that before rank loss, even approximate concealment is bounded below by the minimum gain of the channel, with finite-time lower bounds for regular dynamics that remain positive for all finite times. This complements the recent dynamical resource theory of incompatibility preservability \cite{Hsieh2025} and the classification under classical processing \cite{Das2026}.

Thus, joint measurability of the effective measurements does not imply exact concealment. Adjoint noninjectivity is necessary for even one incompatible pair to be concealable, while a prescribed pair must additionally satisfy the fibre-intersection condition.

The quantitative theory makes this separation operational. Concealment robustness is a generalized robustness on the operational quotient, monotone under kernel inclusion, and reduces to ordinary generalized incompatibility robustness for injective adjoints. Its dual witnesses annihilate inaccessible directions and can be evaluated as shift-invariant normalised positive-payoff channel-output games. This operational interpretation follows the established resource-theoretic paradigm for incompatibility robustness \cite{Regula2018,Takagi2019,Skrzypczyk2019}.

We have completely characterized the qubit channels for which such a retraction exists, extending the unital result to arbitrary qubit channels. The classification reveals that the existence of a positive retraction is strictly stronger than mere kernel noninjectivity: the condition $\mathbf t\in\im(T)$ is necessary for $\rank(T)>0$. This identifies precisely when the affine-fibre problem can be reduced to ordinary joint measurability in the qubit case.

The block-decomposition theorem extends the method to a higher-dimensional class of channels where the adjoint kernel consists of operators with vanishing diagonal blocks. More significantly, the real-symmetric Jordan retraction construction shows that the framework applies to genuinely noncommutative order structures in every dimension, beyond block-diagonal algebras.

The analytical Bloch-rank-two family demonstrates that the quotient changes the resource geometry continuously rather than merely producing the two extreme cases of zero robustness and ordinary incompatibility robustness. The strict separation for $0\le r<1$ therefore reflects genuinely different operational content. Moreover, since concealment implies effective joint measurability, the assemblage generated by the corresponding effective measurement sets admits a local-hidden-state model for every bipartite state \cite{Uola2015}. Thus these measurements cannot certify steering under the channel-restricted access considered here.

The scope is deliberately restricted. We assume finite dimensions, finite-outcome POVMs, exact knowledge of the channel, and tomographically complete input access. Passive finite-dimensional references do not change exact concealment, concealment robustness, or approximate concealment error, even when arbitrary joint compatible representatives and noise POVMs are allowed on the enlarged output space (Section~\ref{sec:ancilla}). Adaptive channel use, genuinely multi-use strategies, arbitrary joint target measurements, and infinite-dimensional systems are not covered. The approximate and finite-sample results provide stability and certification guarantees within the stated single-use model.

Several directions for future work are particularly promising. First, extending the classification of channels admitting positive retractions beyond qubits, block-diagonal channels, and the real-symmetric construction would determine how far the exact reduction extends. Second, obtaining analogous pair-dependent criteria and exact robustness formulas for more general measurement families would strengthen the quantitative theory. Third, extending the framework to active ancilla-assisted target measurements, general channel testers, adaptive protocols, and collective multi-use strategies would clarify whether concealment can be overcome by more general access models. Fourth, a full resource theory of operational concealment with monotones and ordering relations beyond the kernel preorder would provide a more complete quantitative framework, while experimental implementation of the witness construction on near-term quantum platforms would provide a direct test of the operational predictions.

\section{Conclusion}
\label{sec:conclusion}

Our results establish adjoint noninjectivity as the universal structural condition enabling exact concealment of some incompatible binary measurements under channel restrictions. In dynamical settings, the onset of this condition corresponds to rank loss. This mechanism is distinct from contraction, which can render the effective measurements compatible without creating nontrivial operational fibres.

We proved that adjoint noninjectivity is necessary and sufficient for a finite-dimensional channel to conceal some incompatible binary pair, derived quantitative approximate-concealment bounds controlled by the minimum gain of the channel, and established the resulting kernel ordering and finite-time rank-loss obstruction.

The positive-retraction principle provides a dimension-independent reduction to ordinary joint measurability whenever the required map exists. We completely characterized the qubit channels for which such a retraction exists, revealing that unitality is sufficient but not necessary, and that some nonunital rank-deficient channels do not admit such a reduction. We also derived a block-decomposition theorem for higher-dimensional channels and constructed a genuinely non-block real-symmetric Jordan retraction in every finite dimension.

The analytical robustness family and the quotient dual with a shift-invariant normalised positive-payoff operational interpretation reveal a nontrivial channel-dependent resource geometry. These results identify operational concealment as the rank-sensitive layer of measurement incompatibility under channel-restricted state access. We further proved that, for lifted target measurements under a single passive channel use, exact concealment, concealment robustness, and approximate concealment error are invariant under adjoining an arbitrary finite-dimensional reference system, even when arbitrary joint compatible simulators are allowed on the enlarged output space.

\section*{Author Contributions}
M.A.S. and Z.W. conceived the project. M.A.S. developed the analytical and numerical framework, performed the calculations, and prepared the initial draft. Z.W. supervised the project and contributed to the analysis, interpretation of the results, and revision of the manuscript. Both authors reviewed and approved the final manuscript. AI-assisted tools were used for language editing and formatting. All mathematical statements, proofs, and results were verified by the authors, who take full responsibility for the content.

\appendix

\section{Single-state indistinguishability}
\label{app:single_state}

For any finite-dimensional channel $\mathcal E$, any fixed input state $\rho$, and any finite-outcome POVM pair, let $\sigma=\mathcal E(\rho)$. There exist compatible POVMs reproducing the output statistics on $\sigma$. The proof constructs $F_a=p_aI$, $G_b=q_bI$ with $p_a=\Tr[M_a\sigma]$, $q_b=\Tr[N_b\sigma]$, and joint POVM $J_{ab}=p_aq_bI$.

\section{Convexity and compactness of the concealed set}
\label{app:convexity_compactness}

For fixed finite outcome sets $\Omega$ and $\Lambda$, the set $\mathcal C_{\mathcal E}^{\Omega,\Lambda}$ of POVM pairs concealed by $\mathcal E$ is convex and compact.

\emph{Proof of convexity.} If $(M^{(1)},N^{(1)}), (M^{(2)},N^{(2)}) \in \mathcal C_{\mathcal E}$, there exist compatible representatives $(F^{(i)},G^{(i)})$ with $\mathcal E^\dagger(F_a^{(i)}-M_a^{(i)})=0$ and $\mathcal E^\dagger(G_b^{(i)}-N_b^{(i)})=0$. Since compatible POVMs form a convex set, the mixtures $F_a^\lambda = \lambda F_a^{(1)}+(1-\lambda)F_a^{(2)}$, $G_b^\lambda = \lambda G_b^{(1)}+(1-\lambda)G_b^{(2)}$ are compatible and satisfy the required kernel constraints.

\emph{Proof of compactness.} The sets of POVMs with fixed finite outcome sets and of joint POVMs are compact in finite dimension. The feasibility conditions defining concealment are all closed: positivity, normalization, the joint-measurability marginal constraints, and the kernel constraints. Hence the set of all $(M,N,F,G,J)$ satisfying these constraints is closed in a compact set and is therefore compact. The coordinate projection onto $(M,N)$ is continuous, so its image $\mathcal C_{\mathcal E}$ is compact.

\section{Steering consequence}
\label{app:steering}

Quantum steering \cite{Wiseman2007} captures the ability of one party to remotely prepare another party's state through local measurements. Under channel access, Alice's measurement $\{M_a\}$ on a bipartite state $\rho_{AB}$ produces Bob's assemblage. In the Heisenberg picture, the effective measurement on Alice is $\mathcal E^\dagger(M_a)$.

\begin{corollary}[Channel-induced unsteerability under restricted access]
If $\{M_a\}$ and $\{N_b\}$ are concealed by $\mathcal E$, then, for every bipartite state $\rho_{AB}$, the assemblage generated by the effective measurement sets
$\{\mathcal E^\dagger(M_a)\}$ and
$\{\mathcal E^\dagger(N_b)\}$
admits a local-hidden-state model within the single-use channel-restricted scenario.
\end{corollary}

\begin{proof}
By Proposition~\ref{prop:distinction}, concealment implies that the effective measurements are jointly measurable. Jointly measurable measurements generate an unsteerable assemblage for every bipartite state. Hence, by the joint-measurability--steering relation established in Ref.~\cite{Uola2015}, the assemblage produced by the effective measurement sets admits a local-hidden-state model.
\end{proof}

\section{Approximate concealment}
\label{app:approx}

Throughout this appendix, $\|\cdot\|_\infty$ denotes the operator norm. For a linear map $\Phi$, we use the induced norm
\[
\|\Phi\|_{\infty\to\infty}
:=
\sup_{X\ne0}
\frac{\|\Phi(X)\|_\infty}{\|X\|_\infty}.
\]

Exact concealment requires perfect reproduction of all channel-accessible statistics. In practice, finite statistics and imperfect tomography introduce unavoidable errors. We therefore introduce an approximate version.

Given $\varepsilon \ge 0$, we say that $\{M_a\},\{N_b\}$ are \emph{$\varepsilon$-concealed} by $\mathcal{E}$ if there exist compatible POVMs $\{F_a\},\{G_b\}$ such that, for all $a,b$,
\[
\|\mathcal{E}^\dagger(F_a-M_a)\|_\infty \le \varepsilon,\qquad
\|\mathcal{E}^\dagger(G_b-N_b)\|_\infty \le \varepsilon.
\]
The minimal such error defines the concealment error:
\[
\epsilon_c(M,N\mid\mathcal E)
:=\inf_{F,G\ \mathrm{compatible}}
\max\{d_{\mathcal E}(F,M),d_{\mathcal E}(G,N)\},
\]
where $d_{\mathcal E}(F,M):=\max_a\|\mathcal E^\dagger(F_a-M_a)\|_\infty$. Compactness of the compatible-pair set implies that the infimum is attained; consequently, $\epsilon_c=0$ if and only if exact concealment holds.

This approximate formulation admits an exact semidefinite programming representation. The constraints $\|\mathcal{E}^\dagger(F_a - M_a)\|_\infty \le \varepsilon$ are equivalent to
\[
-\varepsilon I \le \mathcal{E}^\dagger(F_a - M_a) \le \varepsilon I,
\]
which are linear matrix inequalities in the joint POVM variables.

\subsection{Approximate ancilla stability}

If $\|\mathcal E^\dagger(M_{a|x}-F_{a|x})\|_\infty \le \varepsilon$, then for every reference system, every state $\rho_{AR}$, and every effect $0\preceq L\preceq I_R$,
\[
\left|\Tr[((M_{a|x}-F_{a|x})\otimes L)(\mathcal E\otimes\id_R)(\rho_{AR})]\right| \le \varepsilon.
\]

To see this, let $\Delta = M_{a|x}-F_{a|x}$. Then
\[
\begin{aligned}
&\left|\Tr[(\Delta\otimes L)(\mathcal E\otimes\id_R)(\rho_{AR})]\right|\\
&\quad= \left|\Tr[(\mathcal E^\dagger(\Delta)\otimes L)\rho_{AR}]\right|\\
&\quad\le \|\mathcal E^\dagger(\Delta)\otimes L\|_\infty\\
&\quad= \|\mathcal E^\dagger(\Delta)\|_\infty\,\|L\|_\infty \le \varepsilon,
\end{aligned}
\]
using tensor-norm multiplicativity and $\|L\|_\infty\le1$. This holds for arbitrary entangled states $\rho_{AR}$.

The stronger $\Gamma_\tau$-based proof in Section~\ref{sec:ancilla} establishes equality of the approximate concealment errors under passive ancillas.

\subsection{Lipschitz stability}

The approximate concealment error is Lipschitz continuous under measurement and channel perturbations.

\begin{theorem}[Lipschitz stability under measurement perturbations]
For two POVM pairs $\mathbf M$ and $\mathbf M'$ with the same outcome sets,
\[
\left|\epsilon_c(\mathbf M\mid\mathcal E)-\epsilon_c(\mathbf M'\mid\mathcal E)\right|
\le
\max_{a,x}\|\mathcal E^\dagger(M_{a|x}-M'_{a|x})\|_\infty.
\]
\end{theorem}

\begin{proof}
Set
\[
\delta_{\mathcal E}
:=
\max_{a,x}
\left\|
\mathcal E^\dagger(M_{a|x}-M'_{a|x})
\right\|_\infty.
\]
Let $\mathbf F=\{F_{a|x}\}$ be any compatible pair that is $\varepsilon$-feasible for $\mathbf M$. Then
\begin{align*}
&
\left\|
\mathcal E^\dagger(F_{a|x}-M'_{a|x})
\right\|_\infty
\\
&\quad\le
\left\|
\mathcal E^\dagger(F_{a|x}-M_{a|x})
\right\|_\infty
+
\left\|
\mathcal E^\dagger(M_{a|x}-M'_{a|x})
\right\|_\infty
\\
&\quad\le
\varepsilon+\delta_{\mathcal E}.
\end{align*}
Taking the infimum over all compatible pairs and all feasible values of $\varepsilon$ gives
\[
\epsilon_c(\mathbf M'\mid\mathcal E)
\le
\epsilon_c(\mathbf M\mid\mathcal E)
+
\delta_{\mathcal E}.
\]
Interchanging $\mathbf M$ and $\mathbf M'$ gives the reverse inequality,
which proves
\[
\left|
\epsilon_c(\mathbf M\mid\mathcal E)
-
\epsilon_c(\mathbf M'\mid\mathcal E)
\right|
\le
\delta_{\mathcal E}.
\]
\end{proof}

\begin{theorem}[Lipschitz stability under channel perturbations]
For channels $\mathcal E$ and $\mathcal F$,
\[
\left|\epsilon_c(\mathbf M\mid\mathcal E)-\epsilon_c(\mathbf M\mid\mathcal F)\right|
\le
\|\mathcal E^\dagger-\mathcal F^\dagger\|_{\infty\to\infty}.
\]
\end{theorem}

\begin{proof}
Let $\mathbf G=\{G_{a|x}\}$ be a compatible pair that is $\varepsilon$-feasible for $\mathbf M$ under $\mathcal E$. Since $G_{a|x}$ and $M_{a|x}$ are effects,
\[
\|G_{a|x}-M_{a|x}\|_\infty\le1.
\]
Therefore,
\begin{align*}
&
\left\|
\mathcal F^\dagger(G_{a|x}-M_{a|x})
\right\|_\infty
\\
&\quad\le
\left\|
\mathcal E^\dagger(G_{a|x}-M_{a|x})
\right\|_\infty
\\
&\qquad+
\left\|
(\mathcal F^\dagger-\mathcal E^\dagger)
(G_{a|x}-M_{a|x})
\right\|_\infty
\\
&\quad\le
\varepsilon+
\left\|
\mathcal F^\dagger-\mathcal E^\dagger
\right\|_{\infty\to\infty}.
\end{align*}
Taking the infimum gives
\[
\epsilon_c(\mathbf M\mid\mathcal F)
\le
\epsilon_c(\mathbf M\mid\mathcal E)
+
\left\|
\mathcal F^\dagger-\mathcal E^\dagger
\right\|_{\infty\to\infty}.
\]
Interchanging $\mathcal E$ and $\mathcal F$ proves the result.
\end{proof}

\section{Finite-sample witness certification}
\label{app:finite_sample}

Once a witness is decomposed as $\mathcal W(\mathbf M)=c_0+\sum_i c_ip_i$, where $c_0$ is known exactly and the $p_i$ are accessible probabilities estimated from independent data blocks, the following finite-sample statement holds.

Let $\widehat p_i$ be estimated from $N_i$ independent Bernoulli trials. Define
\[
\widehat{\mathcal W}=c_0+\sum_i c_i\widehat p_i,\qquad
\Delta_\delta=\sqrt{\frac{\ln(1/\delta)}{2}\sum_i\frac{c_i^2}{N_i}}.
\]
Then, with probability at least $1-\delta$, $\mathcal W(\mathbf M)\ge\widehat{\mathcal W}-\Delta_\delta$. If every $\varepsilon$-concealed tested pair $\mathbf A$ satisfies $\mathcal W(\mathbf A)\le B_\varepsilon$ and $\widehat{\mathcal W}-\Delta_\delta>B_\varepsilon$, then the data certify that $\mathbf M$ is not $\varepsilon$-concealed with confidence at least $1-\delta$.

\begin{proof}
Apply the one-sided Hoeffding inequality~\cite{Hoeffding1963}
to the weighted independent sample averages. Each weighted Bernoulli
contribution has interval length $|c_i|/N_i$, so Hoeffding's inequality
applies with variance proxy $\sum_i c_i^2/N_i$. Since
\[
\widehat{\mathcal W}-\mathcal W(\mathbf M)
=
\sum_i c_i(\widehat p_i-p_i),
\]
Hoeffding's inequality gives
\[
\Pr\!\left[
\widehat{\mathcal W}-\mathcal W(\mathbf M)\ge\Delta_\delta
\right]
\le\delta.
\]
Hence, with probability at least $1-\delta$,
\[
\mathcal W(\mathbf M)\ge\widehat{\mathcal W}-\Delta_\delta.
\]
The certification statement follows immediately from the assumed
bound $\mathcal W(\mathbf A)\le B_\varepsilon$ for every
$\varepsilon$-concealed tested pair $\mathbf A$.
\end{proof}

\section{Exact SDP for concealment robustness}
\label{app:robustness_sdp}

For a pair of finite-outcome POVMs $M=\{M_a\}$ and $N=\{N_b\}$, the concealment robustness is given by
\begin{widetext}
\begin{equation}
\label{eq:robustness_sdp}
\begin{aligned}
R_c(M,N\mid\mathcal E)
=\min_{s,\widetilde J,X,Y,U,V}\quad & s-1 \\[2pt]
\text{s.t.}\quad
& \widetilde J_{ab}\succeq0,
\qquad
  X_a=\sum_b\widetilde J_{ab},
\qquad
  U_b=\sum_a\widetilde J_{ab},\\
& \sum_aX_a=\sum_bU_b=sI,
\qquad
  Y_a\succeq0,
\qquad
  V_b\succeq0,\\
& \sum_aY_a=\sum_bV_b=(s-1)I,\\
& X_a-Y_a-M_a\in\ker(\mathcal E^\dagger),
\qquad \forall a,\\
& U_b-V_b-N_b\in\ker(\mathcal E^\dagger),
\qquad \forall b,\\
& s\ge1.
\end{aligned}
\end{equation}
\end{widetext}

\begin{proposition}
The optimal value of Eq.~\eqref{eq:robustness_sdp} equals $R_c(M,N\mid\mathcal E)$, and the optimum is attained.
\end{proposition}

\begin{proof}
Suppose first that $t\ge0$ is feasible in Definition~\ref{def:robustness}. Thus there are noise POVMs $P=\{P_a\}$, $Q=\{Q_b\}$, compatible representatives $F=\{F_a\}$, $G=\{G_b\}$, and a joint POVM $J=\{J_{ab}\}$ such that, with $s=1+t$,
\[
\begin{aligned}
\mathcal E^\dagger\!\left(
F_a-\frac{M_a+tP_a}{s}
\right)&=0,\\
\mathcal E^\dagger\!\left(
G_b-\frac{N_b+tQ_b}{s}
\right)&=0.
\end{aligned}
\]
Define
\begin{align*}
\widetilde J_{ab}&=sJ_{ab},&
X_a&=sF_a,&
U_b&=sG_b,\\
Y_a&=tP_a,&
V_b&=tQ_b.
\end{align*}
All positivity and normalization constraints follow immediately, and the operational-equivalence relations become
\[
X_a-Y_a-M_a\in\ker(\mathcal E^\dagger),\qquad
U_b-V_b-N_b\in\ker(\mathcal E^\dagger).
\]
Hence every feasible robustness construction gives an SDP point with objective $s-1=t$.

Conversely, let $(s,\widetilde J,X,Y,U,V)$ be feasible for Eq.~\eqref{eq:robustness_sdp}. If $s>1$, define
\begin{align*}
J_{ab}&=\frac{\widetilde J_{ab}}{s},&
F_a&=\frac{X_a}{s},&
G_b&=\frac{U_b}{s},\\
P_a&=\frac{Y_a}{s-1},&
Q_b&=\frac{V_b}{s-1}.
\end{align*}
Then $J$ is a joint POVM with marginals $F,G$, while $P,Q$ are POVMs. Dividing the kernel constraints by $s$ gives
\[
F_a-\frac{M_a+(s-1)P_a}{s}\in\ker(\mathcal E^\dagger),
\]
and analogously for $G_b$. Therefore the noisy pair with parameter $t=s-1$ is concealed.

If $s=1$, positivity together with $\sum_aY_a=\sum_bV_b=0$ implies $Y_a=0$ and $V_b=0$ for every outcome. The operators $J_{ab}=\widetilde J_{ab}$ form a normalized joint POVM with marginals $F_a=X_a$ and $G_b=U_b$, and
\[
F_a-M_a\in\ker(\mathcal E^\dagger),\qquad
G_b-N_b\in\ker(\mathcal E^\dagger).
\]
Thus the original pair is already concealed and $t=0$ is feasible. This treats the endpoint without dividing by $s-1$.

The two constructions preserve the objective and prove equality of the optimal values.

For attainment, the feasible set is nonempty by the finite uniform-noise construction given after Definition~\ref{def:robustness}. Let $s_0<\infty$ be the value obtained from that construction. A minimizing sequence may then be restricted to $1\le s\le s_0$. Since
\[
\sum_aX_a=\sum_bU_b=sI,
\qquad
\sum_aY_a=\sum_bV_b=(s-1)I,
\]
positivity implies
\[
0\preceq X_a,U_b,\widetilde J_{ab}\preceq s_0I,
\qquad
0\preceq Y_a,V_b\preceq(s_0-1)I.
\]
Thus all variables lie in a bounded set. The positivity, normalization, marginal, and kernel constraints are closed, so the restricted feasible set is compact. The continuous objective $s-1$ therefore attains its minimum.
\end{proof}

\section{Technical proofs}
\label{app:technical}

\subsection{Phase-damping derivation}

For the qubit phase-damping channel with rate $\gamma>0$, satisfying $\mathcal E_t^\dagger(\sigma_x)=e^{-\gamma t}\sigma_x$, $\mathcal E_t^\dagger(\sigma_y)=e^{-\gamma t}\sigma_y$, $\mathcal E_t^\dagger(\sigma_z)=\sigma_z$, effective compatibility of planar unbiased binary POVMs with Bloch vectors $\mathbf m,\mathbf n$ in the $(x,y)$ plane occurs iff
\[
e^{-\gamma t}(|\mathbf m+\mathbf n|+|\mathbf m-\mathbf n|)\le2.
\]
For sharp orthogonal measurements in the $(x,y)$ plane, $|\mathbf m+\mathbf n|=|\mathbf m-\mathbf n|=\sqrt2$, yielding $2\sqrt2 e^{-\gamma t}\le2$, hence $t_c=(\log 2)/(2\gamma)$.

\subsection{Singular-value theorem proof}

Let $\mathcal E$ be a unital qubit channel with Bloch matrix $T$ and singular values $\sigma_1\ge\sigma_2\ge\sigma_3$. For a fixed pair of unbiased binary qubit POVMs with Bloch vectors $\mathbf m,\mathbf n$, the effective Bloch vectors are $T^{\mathsf T}\mathbf m$ and $T^{\mathsf T}\mathbf n$. The joint-measurability criterion gives
\[
|T^{\mathsf T}\mathbf m+T^{\mathsf T}\mathbf n|+|T^{\mathsf T}\mathbf m-T^{\mathsf T}\mathbf n|\le2.
\]

For sharp orthogonal pairs, define
\[
\mathbf u=\frac{\mathbf m+\mathbf n}{\sqrt2},\qquad
\mathbf v=\frac{\mathbf m-\mathbf n}{\sqrt2}.
\]
Then $\mathbf u,\mathbf v$ are orthonormal, and
\begin{align*}
&|T^{\mathsf T}(\mathbf m+\mathbf n)|
 +|T^{\mathsf T}(\mathbf m-\mathbf n)|\\
&\quad=
\sqrt2\left(
|T^{\mathsf T}\mathbf u|
+
|T^{\mathsf T}\mathbf v|
\right)\\
&\quad\le
2\sqrt{
|T^{\mathsf T}\mathbf u|^2
+
|T^{\mathsf T}\mathbf v|^2
}\\
&\quad\le
2\sqrt{\sigma_1^2+\sigma_2^2}.
\end{align*}
The upper bound is attained by choosing
$\mathbf m$ and $\mathbf n$ as the two leading left singular
vectors of $T$. For this choice, the transformed sum and
difference have equal norm, and the resulting value is
$2\sqrt{\sigma_1^2+\sigma_2^2}$.
Hence every sharp orthogonal pair is effectively compatible iff $\sigma_1^2+\sigma_2^2\le1$.

\bibliography{references}

\end{document}